\documentclass[a4paper,11pt]{article}
\usepackage{a4wide}
\usepackage{float}
\usepackage{graphicx}
\usepackage[colorinlistoftodos]{todonotes}
\usepackage[colorlinks=true, allcolors=purple]{hyperref}
\usepackage{color}
\usepackage{amsmath}
\usepackage{amsthm}
\usepackage{color,soul}
\usepackage[normalem]{ulem}
\newtheorem{theorem}{Theorem}[section]
\usepackage{xcolor}
\def\sech{\mathop{\rm sech}\nolimits}

\usepackage{comment}

\title{A study of variational single solitary waves governed by the conservative-extended KdV equation with applications to shallow water dispersive shocks }
\author{Saleh Baqer, Hamid Said\\
Department of Mathematics, College of Science, \\
Kuwait University, Sabah Al Salem University City, \\ P.O. Box 5969, Safat 13060, Shadadiya, Kuwait \and}
\date{}

\begin{document}

\maketitle

{\textit{This paper is dedicated to N.F. Smyth, FAustMS and SIAM NWCS member, in loving memory and for his major contributions to the field of applied mathematics of nonlinear dispersive waves in fluids and optics \cite{noelbio}}}

\begin{abstract}
 The extended KdV equation is a nonlinear dispersive wave model that is asymptotically or variationally derived from the full dispersive Euler water waves equations when gravity-capillary and higher order nonlinear effects are taken into account, under weakly nonlinear and long-wave (shallow water limit) approximations. This approximation introduces four additional terms beyond the classical KdV equation: a nonlinear term (cubic nonlinearity), two nonlinear-dispersive terms, and a fully dispersive term (fifth order dispersion). In this paper, we employ a variational approach based on averaged Lagrangians to assess the accuracy and effectiveness of  \textit{single} solitary wave solutions governed by a particular extended KdV equation characterized by energy conservation, within the framework of applications in modulation theory. Compared with solitary wave solutions that have previously been obtained through higher order asymptotics and algebraic methods, the present variational solitonic solutions are notably simpler and more readily applicable to practical problems. Recently, within the rapidly growing field of non-convex dispersive hydrodynamics, the availability of solitary wave solutions to the associated nonlinear dispersive wave models has been recognized as essential for approximating solutions for specific dispersive hydrodynamic phenomena, such as undular bores, commonly referred to as dispersive shock waves. Accordingly, a central objective of this work is to demonstrate the applicability of the derived variational solitary wave solutions to both classical and non-classical (resonant) dispersive shocks, along with the use of the concept of Whitham shocks. Theoretical predictions show excellent agreement with numerical simulations.
\end{abstract}


\section{Introduction}
\label{s:intro}
Understanding the behavior of water waves—encompassing their formation, propagation, and complex interactions—stands at the heart of fluid dynamics and ocean sciences, and more broadly, constitutes a key subject in applied mathematics and physics. The motion of oceanic water waves under the influence of gravity can be reliably described using the incompressible Euler equations, provided they are supplemented with appropriate boundary conditions at the free surface (kinematic and dynamic conditions) and at the bottom boundary (impenetrable condition) \cite{whitham,kamchatnovbook,annabook}. In many situations of interest for modelling water wave propagation, the Euler model can be simplified by assuming that the displacement of the water wave is uniform in the $y-$direction. That is, the fluid elevation $u$ from the undisturbed surface is invariant with respect to the direction transverse to the $x-$direction of the wave propagation. This assumption allows the reduction of the $(3+1)$ dimensional Euler problem to a $(2+1)$ dimensional one. Additionally, assuming that the fluid is irrotational further simplifies the system of equations by allowing the introduction of a velocity potential field. Nevertheless, even under the above assumptions, the Euler system remains a highly nonlinear free surface problem for which no general closed-form solutions are available. This has motivated the development of reductive models by imposing additional physical and mathematical approximations, such as weak nonlinearity, weak dispersion, and the long-wave approximation. These approximations enable the reduction of the full Euler equations to simpler and more tractable shallow water wave models. A particularly important shallow water wave model—and a central focus of the present work—is the Korteweg–de Vries (KdV) equation, which, after non-dimensionalization, takes the form \cite{whitham,kamchatnovbook}
\begin{equation}\label{e:kdv}
u_{t} + 6u u_{x} + u_{xxx} = 0.
\end{equation}

The KdV equation (serving as the unidirectional counterpart of the bidirectional Boussinesq system \cite{whitham,kamchatnovbook}) is of particular interest not only for its physical relevance but also for its rich mathematical structure. It is a completely integrable system, admits exact solutions through the inverse scattering transform method, and possesses an infinite hierarchy of conservation laws \cite{kamchatnovbook}. Furthermore, its applicability extends far beyond shallow water wave theory, arising in a broad range of physically significant contexts, including nonlinear optics \cite{kdvoptics2}, solid mechanics \cite{kdvmechanics1}, Bose–Einstein condensates (BECs) \cite{kdvbec1}, and geophysics \cite{kdvgeophysics}. Among the most prominent nonlinear dispersive wave phenomena observed in these diverse physical settings, as well as in shallow water flows, are solitary waves (commonly termed solitons in the study of integrable systems) and dispersive shock waves (often called undular bores in the context of fluid dynamics), whose dynamics are effectively described by the KdV equation and KdV-type models.

In contrast to classical linear waves, which exist as oscillatory motions, a solitary wave exists as a localized, non-oscillatory hump that propagates over long distances while preserving its shape and speed. Solitary waves were first observed by J.S. Russell in the 1800s along the Edinburgh and Glasgow Union Canal, and their relation to classical wave theory has remained a subject of active research for more than a century. The first detailed mathematical descriptions emerged in the latter half of the nineteenth century, when Boussinesq and Rayleigh independently developed theories for solitary wave evolution. Their analyses confirmed that a solitary wave solution take the form of a $\sech^{2}$ profile vanishing as $|x|\to{\infty}$, and it is as a special case of a cnoidal wave (expressed in terms of the Jacobi elliptic function cn) in the infinite-wavelength, or zero-wavenumber, limit \cite{whitham,kamchatnovbook}. Moreover, the persistence of solitary waves arises from the strong balance between linear (dispersion) and nonlinear effects. The appearance of solitary waves is not confined to hydrodynamics; they appear ubiquitously in diverse physical settings, including optical fibers \cite{agrawal}, liquid crystals \cite{nematicons}, atmospheric flows \cite{solitonatmosp}, and biological systems \cite{biology1}.

A dispersive shock wave (DSW), in contrast, exists as a non-stationary (continuously expanding in space over time) modulated nonlinear periodic wavetrain that connects two distinct initial physical flow states \cite{kamchatnovbook,elreview}. DSWs arise in highly dispersive media where viscosity is negligible or absent, and they can be regarded as the dispersive analogue of classical gas dynamic shocks, in which wave breaking and singularity formation are resolved through the generation of a DSW rather than a steady front as in a classical shock flow. A DSW exhibits a genuinely multi-scale structure, with slowly varying parameters such as the amplitude $a$, wavenumber $k$, frequency $\omega$, and mean height $\bar{u}$, while its wave phase $\theta$ evolves on a fast scale. The leading edge of a DSW is characterized by a solitary wave, whereas its trailing edge consists of a train of small amplitude (linear) harmonic waves. These edges propagate at markedly different velocities, reflecting the strongly nonlinear nature of DSWs; see Figure \ref{f:dsw_regimes}. DSWs were first observed as tsunamis and tidal bores in coastal regions with strong topography, with the Severn River in the UK providing a prominent example, where large waves travel upstream for distances exceeding 20 miles. DSWs also emerge in diverse physical contexts, including quantum fluids \cite{bose1,dswdroplet}, plasmas \cite{plasma3}, magma flows \cite{magmaimplosion}, colloidal media \cite{colloid}, solid mechanics \cite{karimasolidbore1,purohit}, nonlinear optics \cite{fleischer1,fatome}, nonlinear dynamical lattices \cite{patdiscdsw}, atmospheric phenomena (e.g., morning glory clouds) \cite{morning2}, and many others. The evolution, dynamics and solutions of such waves can be extensively analyzed within the mathematical framework of Whitham modulations \cite{whitham,kamchatnovbook,elreview}, which will be discussed in detail in Section \ref{s:app}.
%%%%%%%%%%%%%%%%%%%%%%%%%%%%%%%%%%%%%

Despite the effectiveness of the KdV equation in modelling a variety of physical waves, it has recently been shown to be inadequate for describing several important dispersive hydrodynamic features observed in nature, laboratory experiments, and numerical simulations. Examples include undular bores in solid waveguides induced by fracture \cite{karimasolidbore1}, internal waves in stratified fluids \cite{karimafifth}, surface-tension–dominated water waves in the regime where the Bond number is close to $1/3$ \cite{patkawahara}, the non-monotonic dependence of solitary wave speed on amplitude and the table-top shape of large amplitude solitary waves \cite{solitontable}, and the propagation of linearly polarized optical beams in defocusing (azo-dye doped) nematic liquid crystals as the optical power varies from the low (nonlocal regime) to high (local regime) \cite{salehnem2}. In all such cases, generic nonlinear dispersive wave solutions, such as solitons and dispersive shocks, shed resonant radiation that propagates ahead of the underlying wave. This resonance occurs when the velocity of the leading solitary wave edge matches the phase velocity of the forward-propagating resonant wavetrain. Capturing this non-classical effect requires going one order beyond the classical KdV weakly nonlinear, long-wave approximation and incorporating the next order nonlinear and dispersive terms. This yields the extended KdV (eKdV) equation \cite{ekdv}
\begin{equation}\label{e:ekdv}
 u_{t} + 6 uu_{x} + u_{xxx}
 + \epsilon \left( c_{1}u^{2}u_{x} + c_{2}u_{x}u_{xx} + c_{3} uu_{xxx} + c_{4} u_{xxxxx} \right) = 0,
\end{equation}
where $\epsilon$ measures weak nonlinearity and the coefficients $c_{i}\,(i=1,2,3,4)$ depend on the physical medium through which the wave propagates. For water waves, specifically, $\epsilon$ represents the ratio of wave amplitude to the equilibrium fluid depth, with the higher order coefficients taking the values $c_{1}=-3/2$, $c_{2}=23/4$, $c_{3}=5/2$, and $c_{4}=19/40$ \cite{salehekdv}. The eKdV equation also arises in other contexts, including solitary wave interactions \cite{karimafifth,ekdvinteract}, resonant fluid flow over topography \cite{ekdv}, and ion-acoustic waves in plasmas \cite{ekdvplasma}.

The key distinction between the KdV and eKdV equations—and the fundamental reason for the appearance of resonant radiation—lies in their linear dispersion relations. For the eKdV model, the linear dispersion relation (on background $\bar{u}$) is given by
\begin{equation}\label{e:omegaekdv}
    \omega(k;\bar{u})= (6\bar{u}+\epsilon c_{1}\bar{u}^{2})k-\left(1+\epsilon c_{3}\bar{u}\right)k^3 + \epsilon c_{4} k^5, 
\end{equation}
which can exhibit non-convexity. By contrast, the linear dispersion relation of the classical KdV equation, equation (\ref{e:omegaekdv}) with $c_{i}=0\,(i=1,2,3,4)$, is either purely convex or purely concave. In the non-convex dispersion case, there necessarily exists a turning point $k_{i}$ at which, or in its neighborhood, the qualitative behavior of solitary waves and dispersive shocks markedly changes. Specifically, long- and short-wave components can interact, thereby giving rise to the resonant character of the evolution \cite{patkawahara,nemboreel}.

In essence, resonant radiation causes a decay in both the amplitude and velocity of localized solitary waves and of DSWs. However, its impact is far more pronounced on DSWs than on solitons. Resonant radiation can transform a well-ranked (fully stable) DSW into an ill-ranked (unstable) one, or, in more extreme cases, destroy the entire DSW structure and replace the bore with, for example, a negative-polarity solitary wave \cite{pat}. These transitions give rise to new, non-classical dispersive hydrodynamic phenomena, termed Cross-over Dispersive Shock Waves (CDSWs) and travelling Dispersive Shock Waves (TDSWs). Further details on such non-classical regimes and their mathematical analyses can be found in \cite{patkawahara,pat,patjump,salehnem1,salehekdv,wwp}. It is worth noting that when the higher order coefficients are set to the shallow water wave values, the associated resonant radiation becomes significantly suppressed and rendered nearly unobservable \cite{resekdv, salehekdv}.

As discussed above, solitary wave solutions of the higher order KdV equation (\ref{e:ekdv}) shed resonant radiation, leading to a decay of the underlying wave. This effect is primarily driven by the higher order contributions, most notably the fifth-order derivative term $u_{xxxxx}$ \cite{resekdv,salehekdv}. In this case, the travelling wave profile departs from the classical $\sech^{2}$ structure by developing a small-amplitude Fourier component superimposed on it, which can be systematically analyzed using exponential asymptotic techniques \cite{resekdv,exp1,exp2,exp3}. Recently, however, there has been increasing interest and demand in obtaining solitary wave approximations that remain close to a pure (single) $\sech^{2}$ profile for the higher order KdV equation (\ref{e:ekdv}). Such approximations are particularly valuable in applications involving Whitham modulation theory and in the analysis of a wide class of DSWs. For example, Sprenger \cite{patkawahara} numerically constructed a $\sech^{2}$ solution to investigate the structure of the TDSW regime. Similar approximations have been used to estimate the amplitude (as well as the velocity in case there exists a velocity-amplitude relation) of classical DSWs \cite{equalamp} and to derive solutions for the CDSW regime \cite{salehekdv}. Earlier, Marchant \cite{tim_soliton} employed higher order asymptotic techniques and a nonlocal transformation to obtain approximate solitary wave solutions of the eKdV equation (\ref{e:ekdv}). More recently, Karczewska \textit{et al.} \cite{annabook,annapaper} and Khusnutdinova \textit{et al.} \cite{karimafifth} applied algebraic approaches to construct \emph{single} solitary wave solutions, with further applications to internal wave dynamics. In addition, Khusnutdinova \cite{karimagardner} most recently employed an asymptotic Kodama–Fokas–Liu near-identity transformation to map the eKdV equation to the Gardner equation (equation (\ref{e:ekdv}) with $c_{2}=c_{3}=c_{4}=0$), whose soliton solutions are well established and can thereby be utilized to approximate solitary waves of the eKdV model.

In this manuscript, motivated by the work of \cite{annapaper,karimafifth,equalamp,salehekdv}, we construct, to the best of our knowledge for the first time, a \textit{single} solitary wave solution of the eKdV equation (\ref{e:ekdv}) using the methods of the calculus of variations under the condition $c_{2}=2c_{3}$. This special case is of particular interest, as the eKdV equation (\ref{e:ekdv}) then admits an exact energy conservation law and possesses a corresponding Lagrangian and Hamiltonian structure \cite{annabook,annapaper,karimafifth}. We shall refer to this case as the conservative-extended KdV (conservative-eKdV) model. The derived variational solitary wave solution is notably simpler and more amenable to applications in modulation theory than previously obtained approximations \cite{annabook,annapaper,karimafifth,tim_soliton}. Our variational approach is based on the method of averaged Lagrangians, which has been similarly applied in related contexts such as the nonlinear optics of nematic equations \cite{lagnem1} and BECs \cite{carr}; here, however, we establish its validity through a rigorous variational derivation before applying it. Finally, we demonstrate the utility and effectiveness of the resulting solitary wave solution by applying it to two classes of dispersive hydrodynamic problems: (i) classical DSWs, where the bore is a stable, well-ordered modulated periodic wavetrain without implosive behavior in its interior waves \cite{magmaimplosion}, and (ii) non-classical DSWs, specifically the CDSW regime, in which the leading solitary wave edge of the bore emits resonant radiation and the interior waves of the bore are resonant.

The remainder of this paper is organized as follows. Section \ref{s:model} introduces the mathematical model of interest and formulates the associated Lagrangian and Hamiltonian. In Section \ref{s:varthr}, we apply rigorous variational theory to justify and demonstrate the method used to derive variational single solitary wave solutions. Section \ref{sec:var} employs the method of averaged Lagrangians, developed in the previous section, to construct approximate solitary wave solutions. Section \ref{s:app}  is devoted to applications in both classical and state-of-the-art problems in dispersive hydrodynamics. Finally, Section \ref{s:conc} summarizes our main findings, presents our conclusions, and outlines directions for future research.

\section{Model and the Lagrangian-Hamiltonian formulation}
\label{s:model}

In this paper, we focus on the version of the eKdV equation \eqref{e:ekdv} that admits exact energy conservation laws and possesses both a Lagrangian and a Hamiltonian formulation. As discussed in \cite{annabook,annapaper,karimafifth,salehekdv}, conservation of energy, and as a consequence a Hamiltonian structure, exists only in the special case $c_{2}=2c_{3}$. However, a connection to a Lagrangian formulation has not been made in these past works. The conservative-eKdV model, under the condition $c_{2}=2c_{3}$, takes the form
\begin{equation}\label{e:consekdv}
    u_{t} + 6 uu_{x} + u_{xxx}
 + \epsilon \left( c_{1}u^{2}u_{x} + 2c_{3}u_{x}u_{xx} + c_{3} uu_{xxx} + c_{4} u_{xxxxx} \right) = 0,
\end{equation}
with associated mass and energy conservation laws given by
\begin{equation} \label{e:ekdvmass}
 \frac{\partial u}{\partial t} + \frac{\partial}{\partial x} \left( 3u^{2} + u_{xx} + \epsilon \left[ \frac{1}{3} c_{1} u^{3} + c_{3}uu_{xx} + \frac{1}{2} c_{3} u_{x}^{2} + c_{4} u_{xxxx} \right] \right) = 0,
\end{equation}
and 
\begin{equation}\label{e:energycons}
 \frac{\partial}{\partial t} \left(\frac{u^{2}}{2}\right) + \frac{\partial}{\partial x} \left( 2u^{3} + uu_{xx}
 - \frac{1}{2}u_{x}^{2} + \epsilon \left[ \frac{1}{4} c_{1} u^{4} + c_{3} u^{2}u_{xx} + c_{4} uu_{xxxx} - c_{4} u_{x}u_{xxx} + \frac{1}{2} c_{4} u_{xx}^{2} \right] \right) = 0,
\end{equation}
respectively. 

Analogous to the classical KdV equation (\ref{e:kdv}), a Lagrangian for the conservative-eKdV equation \eqref{e:consekdv} can be obtained by introducing a potential function $u=\phi_{x}$ \cite{whitham,kamchatnovbook}. Substituting this representation yields the equivalent formulation
\begin{equation}\label{phi_consekdv}
\phi_{xt}+6\phi_{x}\phi_{xx}+\phi_{xxxx}+\epsilon\left(c_{1}\phi_{x}^2\phi_{xx}+ 2c_3 \phi_{xx} \phi_{xxx} + c_3\phi_x \phi_{xxxx}  +c_{4}\phi_{xxxxxx}\right)=0,
\end{equation}
with the associated Lagrangian density 
\begin{equation} \label{Lag}
    L= \frac{1}{2}\phi_{t}\phi_{x} + \phi^{3}_{x} - \frac{1}{2}\phi^2_{xx} + \epsilon\left( \frac{1}{12} c_{1} \phi^4_{x} + \dfrac{1}{2}c_3 \phi_x \phi^2_{xx}  + \dfrac{1}{2}c_3 \phi_x^2 \phi_{xxx}  + \frac{1}{2} c_{4} \phi^2_{xxx}\right).
\end{equation}
This Lagrangian satisfies the Euler–Lagrange equation
\begin{equation}
    \frac{\delta \mathcal{L}}{\delta \phi}=-\frac{\partial}{\partial t}(L_{\phi_{t}})-\frac{\partial}{\partial x}(L_{\phi_{x}})+\frac{\partial^2}{\partial x^2}(L_{\phi_{xx}})-\frac{\partial^3}{\partial x^3}(L_{\phi_{xxx}})=0,
\end{equation}
which recovers equation \eqref{phi_consekdv}. Expressed in terms of the travelling wave phase $\theta = x - V_{s} t$, where $V_{s}$ denotes the solitary wave velocity, the Lagrangian takes the form
\begin{equation} \label{consekdv_L}
    L = -\frac{1}{2}V_{s}\phi_{\theta}^2 + \phi^{3}_{\theta} - \frac{1}{2}\phi^2_{\theta \theta} + \epsilon\left(  \frac{1}{12}c_{1} \phi^4_{\theta} + \dfrac{1}{2}c_3 \phi_\theta \phi^2_{\theta \theta}  + \dfrac{1}{2}c_3 \phi_\theta^2 \phi_{\theta \theta \theta} + \frac{1}{2} c_{4} \phi^2_{\theta \theta \theta} \right).
\end{equation}

We now establish the corresponding Hamiltonian formulation to the conservative-eKdV equation \eqref{phi_consekdv}. As in the classical KdV case, the Hamiltonian can be shown to be a constant of motion. Starting from the Lagrangian \eqref{Lag}, we perform the Legendre transform 
\begin{equation} \label{leg}
    H = p \cdot \phi_t - L,
\end{equation}
where $p = \dfrac{\partial L}{\partial \phi_t}$. In terms of the potential $\phi$ and its derivatives, the Hamiltonian becomes
\begin{equation*}
    H = - \phi_x^3 + \dfrac{1}{2} \phi_{xx}^2 - \epsilon \left( \dfrac{c_1}{12} \phi_x^4 + \dfrac{c_3}{2}\phi_x \phi_{xx}^2 + \dfrac{c_3}{2} \phi_x^2 \phi_{3x}+ \dfrac{c_4}{2} \phi_{3x}^2 \right),
\end{equation*} 
and satisfies the relation
\begin{align} \label{L+H}
  \left(  \dfrac{\partial }{\partial \phi} - \dfrac{\partial }{\partial x} \left( \dfrac{\partial }{\partial \phi_x} \right) + \dfrac{\partial^2 }{\partial x^2} \left( \dfrac{\partial }{\partial \phi_{xx}} \right) - \dfrac{\partial^3 }{\partial x^3} \left( \dfrac{\partial }{\partial \phi_{xxx}} \right) \right) (L+H) = 0.
\end{align}
Since $\dfrac{\partial L}{\partial \phi_t} = \tfrac{1}{2} \phi_x$, we conclude that
\begin{equation}
    \phi_{xt} = \dfrac{\delta \mathcal{H}}{\delta \phi},
\end{equation}
which, in terms of the original variable $u$, yields an explicit Hamiltonian structure:
\begin{equation} \label{ham-eq}
    u_t = \mathcal{J} \cdot \dfrac{\delta \mathcal{H}}{\delta u},
\end{equation}
with operator $\mathcal{J} = \partial_x$ and Hamiltonian functional $\mathcal{H}$ defined by
\begin{equation} \label{H-u}
      \mathcal{H}(u) =  \int_{- \infty}^{\infty} H(u, u_x, u_{xx}) dx =  \int_{- \infty}^{\infty} \left( - u^3 + \dfrac{1}{2} u_{x}^2 - \epsilon \left( \dfrac{c_1}{12} u^4 + \dfrac{c_3}{2} u u_{x}^2 + \dfrac{c_3}{2} u^2 u_{x x}+ \dfrac{c_4}{2} u_{ x x}^2 \right) \right) dx.
\end{equation}
Then, it follows that
\begin{equation} \label{H-cons}
    \int_{- \infty}^{\infty} \left( - u^3 + \dfrac{1}{2} u_{x}^2 - \epsilon \left( \dfrac{c_1}{12} u^4 + \dfrac{c_3}{2} u u_{x}^2 + \dfrac{c_3}{2} u^2 u_{x x}+ \dfrac{c_4}{2} u_{ x x}^2 \right) \right) dx = \mathrm{constant} \, .
\end{equation}
This agrees with previous results, though a different approach (called \emph{near identity transformations}) was employed to produce the same result \cite{annabook,annapaper,karimafifth}. We, on the other hand, have obtained the Hamiltonian formulation naturally via the Legendre transformation \eqref{leg} without resorting to any approximation scheme\footnote{In fact the term $-\frac{1}{2}c_{3} u^2 u_{xx}$ in the Hamiltonian \eqref{H-u} does not appear in these past works. However, one can produce it via integration by parts. This could explain why the connection to a Lagrangian had not been made previously.}. 

\section{Variational theory}
\label{s:varthr}
Let us begin by considering a functional of the form
\begin{equation} \label{lagrangian}
    \mathcal{L}(\eta) = \int_{-a}^a L(\eta, \eta ', \theta) \, d \theta, 
\end{equation}
where the filed $ \eta $ is a function of a single variable $\theta$. Essential to the calculus of variations \cite{fomin} is the study of the extrema of functional $\mathcal{L}$ subject to boundary conditions $\eta(a) = \alpha$ and $\eta(-a) = \beta $. In other words, we look to extremize functional \eqref{lagrangian} belonging to the admissible set 
\begin{equation} \label{admissible}
    \Omega = \lbrace \eta \in C^2[-a, a]: \eta(a) = \alpha, \, \eta(-a) = \beta \rbrace.
\end{equation}
Choose a smooth function satisfying $h(a) = h(-a) = 0$, and define the function

\begin{equation} \label{l-tau}
l (\tau) =\mathcal{L} (\eta_\tau),
\end{equation}
where $\eta_\tau \doteq \eta + \tau h$ belongs to the admissible set $\Omega$. If $\eta$ extremizes the functional, then it is achieved at $\tau = 0$, which implies $l'(0) = 0$.
Now introduce the following ansatz: 

\begin{equation} \label{ansatz}
\eta = \eta(A(\theta), w(\theta), \theta),
\end{equation}
 where $A$ and $w$ are two functions (e.g. wave parameters) that need to be determined. Under the above ansatz, we can formulate the \emph{averaged Lagrangian} $\overline{\mathcal{L}}$: 
 
 \begin{equation} \label{av}
\overline{\mathcal{L}} (A,w) = \mathcal{L}\circ\eta \, (A,w) = \int_{-a}^{a} L( \eta (A,w), \eta'(A,w), \theta) \, d \theta.
\end{equation}

We similarly define $A_\tau = A + \tau \varphi$ and $w_\tau = w + \tau \psi$, where $\varphi$ and $\psi$ are two smooth functions with support over $[-a, a]$, and assume that dependence on the (first) derivatives of $A$ and $w$ may appear in the argument $\eta'(A,w)$. Here again, at $\tau=0$ we obtain the extremum $\eta(A,w)$, therefore
\begin{equation} \label{l_av}
   \overline{l}'(0)  = 0,
\end{equation}
where $\overline{l}(\tau) \doteq \overline{\mathcal{L}} (A_\tau,w_\tau)$. Now we carry forward the calculation implied by equation \eqref{l_av}:
\begin{align*}
    \overline{l}'(0) &= \int_{-a}^{a} \dfrac{d}{d\tau} \left[ L( \eta (A_\tau,w_\tau), \eta'(A_\tau,w_\tau), \theta)\right]_{\tau = 0} \, d \theta \\
    &= \int_{-a}^{a} \left[ \dfrac{\partial L}{\partial \eta} \dfrac{\partial \eta} {\partial A_\tau} \dfrac{d A_\tau}{d \tau} + \dfrac{\partial L}{\partial \eta'} \dfrac{\partial \eta'} {\partial A_\tau} \dfrac{d A_\tau}{d \tau} + \dfrac{\partial L}{\partial \eta'} \dfrac{\partial \eta'} {\partial A'_\tau} \dfrac{d A'_\tau}{d \tau} \right. \\
    &\mathrel{\phantom{=}}  \qquad \left. + \dfrac{\partial L}{\partial \eta} \dfrac{\partial \eta} {\partial w_\tau} \dfrac{d w_\tau}{d \tau}  + \dfrac{\partial L}{\partial \eta'} \dfrac{\partial \eta'} {\partial w_\tau} \dfrac{d w_\tau}{d \tau} + \dfrac{\partial L}{\partial \eta'} \dfrac{\partial \eta'} {\partial w'_\tau} \dfrac{d w'_\tau}{d \tau}\right]_{\tau = 0}d \theta \\
    &= \int_{-a}^a \left [ \dfrac{\partial L}{\partial \eta} \dfrac{\partial \eta} {\partial A} \varphi + \dfrac{\partial L}{\partial \eta'} \dfrac{\partial \eta'} {\partial A} \varphi + \dfrac{\partial L}{\partial \eta'} \dfrac{\partial \eta'} {\partial A'} \varphi' + \dfrac{\partial L}{\partial \eta} \dfrac{\partial \eta} {\partial w} \psi + \dfrac{\partial L}{\partial \eta'} \dfrac{\partial \eta'} {\partial w} \psi + \dfrac{\partial L}{\partial \eta'} \dfrac{\partial \eta'}{\partial w'} \psi' \right ] d \theta \\
    &= \int_{-a}^a \left [ \left( \dfrac{\partial L}{\partial A} - \dfrac{d}{d \theta} \left( \dfrac{\partial L}{\partial A'} \right) \right) \varphi   + \left( \dfrac{\partial L}{\partial w} - \dfrac{d}{d \theta} \left( \dfrac{\partial L}{\partial w'} \right) \right) \psi  \right ] d \theta\\
    &=0,
\end{align*}
and since this equality hold for all test functions $\varphi$ and $\psi$, we conclude 
\begin{equation}\label{Av_EL}
  \left\{
  \begin{array}{ll}
  \vspace*{0.1 in}
    \dfrac{\partial L}{\partial A} - \dfrac{d}{d \theta} \left( \dfrac{\partial L}{\partial A'} \right) =0, \\
    
    \dfrac{\partial L}{\partial w} - \dfrac{d}{d \theta} \left( \dfrac{\partial L}{\partial w'} \right) = 0.\\
    \end{array}
     \right.
\end{equation}
The above argument can be extended, in a straightforward manner, if the ansatz in \eqref{ansatz} is assumed to hold for $n$ independent parameters instead of only 
two. We have, hence, proved the following

\begin{theorem}\label{thm:1}
   Let $\eta \in \Omega$ be an extremum of functional \eqref{lagrangian}, and assume the extremum depends on a set of $n$  independent functions $\eta = \eta(z_1 (\theta), z_2(\theta), ..., z_n(\theta))$. Then $z_i$ solves the Euler-Lagrange equations:
   \begin{displaymath}
  \dfrac{\delta \overline{\mathcal{L}}}{\delta z_i} = \dfrac{\partial L}{\partial z_i} - \dfrac{d}{d \theta} \left( \dfrac{\partial L}{\partial z_i'} \right) = 0
   \end{displaymath}
for each $i = 1,2, \ldots, n$.
\end{theorem}

We can now generalize the previous result as follows. Consider
\begin{equation} \label{m_lagrangian}
    \mathcal{L}(\eta) = \int_{-a}^a L(\eta, \eta ', \eta '', \ldots, \eta^{(m)}, \theta) \, d \theta 
\end{equation}
and the admissible set 
\begin{equation} \label{m_admissible}
    \Omega = \lbrace \eta \in C^{\, m+1}[-a, a]: \eta(a) = \alpha_0, \, \eta(-a) = \beta_0, \ldots , \eta^{(m-1)}(a) = \alpha_{m-1}, \, \eta^{(m-1)}(-a) = \beta_{m-1} \rbrace,
\end{equation}
where the boundary pairs $(\alpha_0, \beta_0), \ldots, (\alpha_{m-1}, \beta_{m-1})$ are fixed. Then we have the following theorem: 

\begin{theorem}\label{thm:2}
  Let $\eta \in \Omega$ be an extremum of functional \eqref{m_lagrangian}, and assume the extremum depends on a set of $n$  independent functions $\eta = \eta(z_1 (\theta), z_2(\theta), ..., z_n(\theta))$. Then $z_i$ solves the Euler-Lagrange equations:
  \begin{displaymath}
   \dfrac{\delta \overline{\mathcal{L}}}{\delta z_i} = \dfrac{\partial L}{\partial z_i} - \dfrac{d}{d \theta} \left( \dfrac{\partial L}{\partial z_i'} \right) + \dfrac{d^2}{d \theta^2} \left( \dfrac{\partial L}{\partial z_i''} \right) - \ldots + (-1)^m \dfrac{d^m}{d \theta^m} \left( \dfrac{\partial L}{\partial z_i^{(m)}} \right)  = 0
  \end{displaymath}
for each $i = 1,2,\ldots, n$.
\end{theorem}
\begin{proof}
  Let $\overline{l}(\tau) = \overline{\mathcal{L}} (\vec{z}_\tau)$, then the relation $\overline{l}'(0) = 0$ implies
  \begin{displaymath}
    \int_{-a}^a \left( \displaystyle\sum_{i=1}^n \displaystyle\sum_{k=0}^m \displaystyle\sum_{j=k}^m \dfrac{\partial L}{\partial \eta^{(j)}} \dfrac{\partial \eta^{(j)}}{\partial z_i^{(k)}} \varphi_i^{(k)}  \right) d\theta = 0,
  \end{displaymath}
  where $\varphi_1, ..., \varphi_n $ are a collection of $n$ dependent smooth functions with upport over $[-a,a]$. On the other hand, we have 
  \begin{displaymath}
    \displaystyle\sum_{j=k}^m \dfrac{\partial L}{\partial \eta^{(j)}} \dfrac{\partial \eta^{(j)}}{\partial z_i^{(k)}} =  \dfrac{\partial L}{\partial z_i^{(k)}}
  \end{displaymath}
for every $k=0, 1,\ldots, m$ and $i =1, 2,\ldots,n$. Combining the above two equations yields the the desired result. 
\end{proof}

\section{Variational solitary wave solutions}\label{sec:var}

In this section, we derive approximate solitary wave solutions using the averaged Lagrangian method described above. Specifically, we seek a traveling solitary wave solution of the form
\begin{equation} \label{soliton_Riesz_general}   u=\phi_{\theta}=a_{s}\mathrm{sech}^2(w_{s}\,\theta_{s}); \quad a_{s}\,(\text{{fixed parameter}}),
\end{equation}
where $\theta_{s}=x-V_{s}t$ is the solitary wave phase, and $w_{s}$ denotes the inverse width of the solitary wave. We assume that derivatives of all wave parameters are negligibly small and therefore can be ignored, as is standard when applying average variational principles. Here, the parameter $a_{s}$ denotes a (fixed) solitary wave amplitude, which is related to the \textit{actual (higher order)} solitary wave height $A_{s}$. The latter will be determined later using the energy method. Expressions for the wave parameters $V_{s}$ and $w_{s}$ are obtained by averaging the Lagrangian (\ref{consekdv_L}) and applying the corresponding variational procedure.

To proceed, we substitute the trial solution (\ref{soliton_Riesz_general}) into the Lagrangian (\ref{consekdv_L}) and average over the solitary wave phase, $-\infty<\theta_{s}<\infty$. This yields the averaged Lagrangian
\begin{equation}\label{averaged_L}
    \overline{\mathcal{L}} = \int^{\infty}_{-\infty} L\,d\theta = \frac{16 }{15}\frac{a^3_{s}}{w_{s}} - \frac{8}{15}w_{s}a^2_{s} - \frac{2}{3}\frac{V_{s}a^2_{s}}{w_{s}} + \epsilon \left(\frac{8}{105}\frac{a^{4}_{s}}{w_{s}}c_{1} -\dfrac{32}{105} a_{s}^3 w_s c_3 +\frac{32}{21}w_{s}^3a^2_{s}c_{4}\right) . 
\end{equation}
Taking the variation of (\ref{averaged_L}) with respect to $a_{s}$ gives
\begin{equation}\label{La}
\delta{a_{s}}:\quad \dfrac{\delta \overline{\mathcal{L}}}{\delta a_{s}}=\frac{16}{15}\frac{a^{2}_{s}}{w_{s}}-\frac{4}{3}\frac{a_{s}V_{s}}{w_{s}}-\frac{16}{15}a_{s}w_{s}+\epsilon\left(\frac{32}{105}\frac{a^{3}_{s}}{w_{s}}c_{1}  -\dfrac{32}{35} a_{s}^2 w_s c_3 +\frac{64}{21}a_{s}w^{3}_{s} c_{4}\right)=0 ,
\end{equation}
while variation with respect to $w_{s}$ yields
\begin{equation}\label{Lw}
\delta{w_{s}}:\quad \dfrac{\delta \overline{\mathcal{L}}}{\delta w_{s}}=\frac{2}{3}\frac{a^2_{s}V_{s}}{w^{2}_{s}}-\frac{16}{15}\frac{a^{3}_{s}}{w^{2}_{s}}-\frac{8}{15}a^{2}_{s}+\epsilon\left(-\frac{8}{105}\frac{a^{4}_s}{w^{2}_{s}}c_{1} -\dfrac{32}{105} a_{s}^3 c_3 + \frac{32}{7}w^{2}_{s}a^{2}_{s}c_{4}\right)= 0. 
\end{equation}
From equation (\ref{La}), an explicit expression for the solitary wave velocity $V_{s}$ follows, namely
\begin{equation}\label{Vs_1}
    V_{s}=\frac{12}{5} a_{s} - \frac{4}{5} w^2_{s} + \epsilon\left(\frac{8}{35}a^{2}_{s}c_{1} -\dfrac{24}{35} a_{s} w_s^2 c_3 +\frac{16}{7}w^{4}_{s}c_{4} \right). 
\end{equation}
Substituting this velocity expression into equation (\ref{Lw}) produces four possible roots for the inverse width $w_{s}$,
\begin{eqnarray}
    w_{s} &=& \pm \frac{1}{4}\sqrt{\frac{7 + 5a_{s}c_3\epsilon  + \sqrt{49-560a_{s}c_{4}\epsilon-80a^{2}_{s}c_{1}c_{4}\epsilon^{2} + 25a_{s}^2 c_3^2 \epsilon^2 + 70 a_{s}c_3\epsilon  }}{5\epsilon c_{4}}}, \\
    w_{s} &=& \pm \frac{1}{4}\sqrt{\frac{7 + 5a_{s}c_3\epsilon  - \sqrt{49-560a_{s}c_{4}\epsilon-80a^{2}_{s}c_{1}c_{4}\epsilon^{2} + 25a_{s}^2 c_3^2 \epsilon^2 + 70 a_{s}c_3\epsilon  }}{5\epsilon c_{4}}}, \label{eq:right_ws}
\end{eqnarray}
where the negative roots are discarded due to their incompatibility with the standard KdV limit $c_{1}=c_{3}=c_{4}=0$ \cite{whitham}. To select the physically correct width between the two remaining positive roots, we take the limit $\epsilon \to 0$, which reduces the conservative-eKdV equation (\ref{e:consekdv}) to the classical KdV equation (\ref{e:kdv}). This identifies the positive root in (\ref{eq:right_ws}) as the appropriate width, namely,
\begin{equation}
    \lim_{\epsilon\to{0}} \frac{1}{4}\sqrt{\frac{7 + 5a_{s}c_3\epsilon  - \sqrt{49-560a_{s}c_{4}\epsilon-80a^{2}_{s}c_{1}c_{4}\epsilon^{2} + 25a_{s}^2 c_3^2 \epsilon^2 + 70 a_{s}c_3\epsilon  }}{5\epsilon c_{4}}} = \sqrt{\frac{a_{s}}{2}}.
\end{equation}
Since the conservative-eKdV equation (\ref{e:consekdv}) is a correction to the classical KdV at order $\mathcal{O}(\epsilon)$, it is natural to expand the inverse width expression (\ref{eq:right_ws}) in a Taylor series up to this asymptotic order. Doing so in (\ref{eq:right_ws}) yields
\begin{equation}\label{eq:ws_per}
\boxed{
    w_{s}=\sqrt{\frac{a_{s}}{2}}+\epsilon a^{3/2}_{s}\left(\frac{\sqrt{2}}{28}c_{1} - \dfrac{5 \sqrt{2}}{28} c_3+\frac{5\sqrt{2}}{7}c_{4}\right)+\mathcal{O}(\epsilon^2),}
\end{equation}
which, when substituted into (\ref{Vs_1}), gives the corrected velocity
\begin{equation}\label{eq:vs_per}
\boxed{
    V_{s}=2a_{s}+\epsilon a^{2}_{s}\left(\frac{6}{35}c_{1} - \dfrac{2}{35} c_3-\frac{4}{7}c_{4}\right)+\mathcal{O}(\epsilon^2).}
\end{equation}
Here, $a_{s}$ remains fixed and independent of $\epsilon$, as previously stressed. At this point, we have derived variational solutions for the velocity and inverse width of a single solitary wave governed by the conservative-eKdV equation. However, the presence of higher order nonlinear terms in the conservative-eKdV equation is expected to perturb the solitonic height $A_s$ as well. To capture this, additional steps are required to determine $A_{s}$ in terms of the fixed amplitude parameter $a_s$, namely we employ the energy method, as done in \cite{whitham,patkawahara}.

\subsection{Conservative-eKdV velocity-amplitude relation and solitonic height}

To calculate the actual (higher order) solitary wave height $A_{s}$, we derive first a velocity-amplitude relation using the energy method. Following \cite{whitham,patkawahara}, let $u=u(\theta_{s})$ be a travelling wave solution. Then the general eKdV equation (\ref{e:ekdv}) can be written as
%\begin{equation}
 %   -V_{s}u'+3(u^2)'+u'''+\epsilon\left(\frac{c_1}{3}(u^3)'+\frac{c_3}{2}(u'^2)'+c_3(uu'')'+c_{4}u^{(5)}\right)=0,
%\end{equation}
\begin{equation}
    -V_{s}u'+3(u^2)'+u'''+\varepsilon\left(\frac{c_1}{3}(u^3)'+\frac{c_2}{2}(u'^2)'+c_3uu'''+c_{4}u^{(5)}\right)=0.
\end{equation}
Integration yields 
%\begin{equation}\label{prior_energy}
 %   -V_{s}u + 3(u^2) + u''+\epsilon\left(\frac{c_1}{3}(u^3) +\frac{c_3}{2}(u'^2) +c_3(uu'') +c_{4}u^{(4)}\right)=E_{1},
%\end{equation}
\begin{equation}\label{prior_energy}
    -V_{s}u+3u^2+u''+\frac{\varepsilon c_1}{3}u^3 + \frac{\varepsilon}{2} (c_2-c_3) u'^2 + c_3\varepsilon u u'' +c_4\varepsilon u^{(4)}=E_{1},
\end{equation}
where $E_{1}$ is an integration constant determined by the solitonic boundary condition 
\begin{equation}
    \lim_{\theta_{s}\to{-\infty}}u=0 \quad \implies \quad E_{1}=0.
\end{equation}
Multiplying equation (\ref{prior_energy}) by $u'$ and integrating produces the so-called energy estimate equation \cite{patkawahara}
\begin{equation}\label{energy_integral_1}
    -\frac{V_{s}}{2}u^2+u^3+\frac{1}{2}u'^2+\frac{\varepsilon c_1}{12}u^4 + \varepsilon c_4 \left(u'u''' -\frac{1}{2}u''^2\right) + \frac{\varepsilon (c_2-c_3)}{2} \int u'^3  + \varepsilon c_{3} \int u u' u''  = E_{2},
\end{equation}
with a constant $E_{2}$ that is again determined by the solitonic boundary condition 
\begin{equation}
    \lim_{\theta_{s}\to{-\infty}}u=0 \quad \implies \quad E_{2}=0.
\end{equation}
Using integration by parts, one can rewrite the integrals in (\ref{energy_integral_1}) in terms of each other, e.g.,
\begin{multline}\label{energy_integral_2}
    -\frac{V_{s}}{2}u^2+u^3+\frac{1}{2}u'^2+\frac{\varepsilon c_1}{12}u^4 + \varepsilon c_4 \left(u'u''' -\frac{1}{2}u''^2\right) \\ + \frac{\varepsilon}{2}(c_2-c_3)uu'^2  
    + \varepsilon(2c_{3}-c_{2})\int u u' u'' = 0. 
\end{multline}
In the special case $c_{2}=2c_{3}$, this reduces (\ref{energy_integral_2}) to
\begin{equation}
    -\frac{V_{s}}{2}u^2+u^3+\frac{1}{2}u'^2+\frac{\varepsilon c_1}{12}u^4 + \varepsilon c_4 \left(u'u''' -\frac{1}{2}u''^2\right) + \frac{\varepsilon}{2}c_{3}uu'^2 =0.
\end{equation}
Substituting the solitonic ansatz, with the actual solitonic  height $A_{s}$,
\begin{equation}\label{e:ansatz}
    u=A_{s}\sech^{2}(w_{s}\theta_{s}),
\end{equation}
into the above expression and evaluating at the soliton maximum $\theta_{s}=0$ gives
\begin{equation}\label{energy_max}
    -\frac{V_{s}}{2}u^2_{\text{max}}+u^3_{\text{max}}+\frac{1}{2}u'^2_{\text{max}}+\frac{\varepsilon c_1}{12}u^4_{\text{max}} + \varepsilon c_4 \left(u'_{\text{max}}u'''_{\text{max}} -\frac{1}{2}u''^2_{\text{max}}\right) + \frac{\varepsilon}{2}c_{3}u_{\text{max}}u'^2_{\text{max}}=0,
\end{equation}
where
\begin{equation}
    u^{j}_{\text{max}}=A^{j}_{s},\quad u''^2_{\text{max}}=4A^{2}_{s}w^{4}_{s}, \quad u'^2_{\text{max}}=u'_{\text{max}}u'''_{\text{max}}=u_{\text{max}}u'^2_{\text{max}}=0,
\end{equation}
with the exponent $j\in{\left\{1,2,3,4\right\}}$. Equation (\ref{energy_max}) yields four possible values for the solitonic height $A_{s}$ in the ansatz (\ref{e:ansatz}). Of these, two correspond to trivial (zero) roots, one is negative—discarded here, since we restrict our attention to positive-polarity (bright or elevation) solitary waves that are relevant for the modulation theory applications considered in this paper—and one is positive, namely
\begin{equation}
A_{s}=\frac{-6+\sqrt{36+6V_{s}c_{1}\epsilon+24c_{1}c_{4}w^{4}_{s}\epsilon^2}}{\epsilon c_{1}}. 
\end{equation}
Expanding this positive root in a Taylor series up to order $\mathcal{O}(\epsilon)$ gives
\begin{equation}
    A_{s}=\frac{V_{s}}{2}+\epsilon\left(-\frac{1}{48}c_{1}V^{2}_{s}+2c_{4}w^{4}_{s}\right). 
\end{equation}
Finally, substituting the variational solutions (\ref{eq:ws_per}) and (\ref{eq:vs_per}) allows us to express the solitonic height $A_{s}$ in terms of the fixed amplitude parameter $a_{s}$ as
\begin{equation}
\boxed{
    A_{s}=a_{s}+\epsilon a^{2}_{s}\left(\frac{1}{420}c_{1}  -\frac{1}{35} c_3+\frac{3}{14}c_{4}\right)+\mathcal{O}(\epsilon^{2}).}
    \label{e:Hs_blue}
\end{equation}
Note that when $\epsilon=0$ the parameter $a_{s}$ serves as the solitonic height for the classical KdV case. 

\section{Applications to dispersive shocks}
\label{s:app}

In this section, we present applications of the above variational solitary wave solutions to dispersive shocks. Before doing so, we briefly review the fundamental methods used to analyze DSWs. 

In general, solutions describing DSW evolutions can be derived within the framework of modulation theory \cite{whitham,kamchatnovbook,elreview}, originally formulated by G.B. Whitham in the mid-1960s. The central aim of Whitham’s theory is to provide a mathematical description of how linear and nonlinear wavetrains evolve under slow variations in space and time, while also addressing their stability under the influence modulations. Three principal approaches exist for the development of Whitham’s theory: the method of averaged conservation laws, the method of multiple scales in perturbation theory, and the method of averaged Lagrangians \cite{whitham,kamchatnovbook}. Although distinct in formulation, all of these approaches lead to equivalent conclusions; however, the averaged Lagrangian method yields an additional conservation law, namely the conservation of wave action. Through Whitham’s averaging procedure, one obtains a system of partial differential equations (PDEs) governing the slow evolution of wave parameters such as amplitude $a$, wavenumber $k$, frequency $\omega$, and mean level $\bar{u}$. These PDEs are known as modulation equations, and the complete set is referred to as the modulation system. Whitham further demonstrated that the nature of the modulation system determines wave stability: if the system is hyperbolic (with real eigenvalues and a complete set of eigenvectors), the corresponding wave is stable; if the system is elliptic (with complex eigenvalues), the wave is unstable \cite{whitham,kamchatnovbook}.

The next major breakthrough following the development of Whitham modulation theory was achieved by Pitaevskii \textit{et al} \cite{plasmapit}. They demonstrated that when a modulation system can be expressed in Riemann invariant form, its expansion fan (or simple wave) solution directly yields the structure of a DSW. Pitaevskii’s seminal work focused on DSWs governed by the classical KdV equation (\ref{e:kdv}). A subsequent advance established, using rigorous tools from functional analysis, that all integrable systems admit a Riemann invariant representation \cite{analysis}. Consequently, DSW solutions for integrable models are guaranteed to exist. This insight led to the formulation of explicit DSW solutions for a broad class of integrable nonlinear dispersive wave equations, including the defocusing nonlinear Schr\"{o}dinger (NLS) equation, among others \cite{elreview}. 

In most real-world applications, however, the governing models are non-integrable, so a Riemann-invariant formulation—and therefore a guaranteed DSW solution—cannot be assumed. This presents a central challenge: developing methods to construct DSW solutions for non-integrable nonlinear dispersive hydrodynamic models. Two major contributions addressing this challenge are (i) the DSW fitting method (introduced by El) \cite{fitting} and (ii) the DSW equal amplitude approximation (developed by Marchant and Smyth) \cite{equalamp}. El’s DSW fitting method relies on the observation that the modulation equations become degenerate at the trailing and leading edges of a DSW and adopt a standard structure that can be identified \textit{without explicit knowledge} of the modulation equations themselves. Remarkably, the only essential ingredient in this approach is the linear dispersion relation $\omega(k;\bar{u})$ associated with the hydrodynamic model under consideration. This enables the determination of key macroscopic DSW properties, including the amplitudes and velocities at its edges. However, the method is limited to stable, KdV-type DSWs. To overcome this limitation and approximate solutions for both stable and unstable DSWs, regardless of the integrability property as well, Marchant and Smyth introduced the DSW equal amplitude approximation, which plays a central role in the present work. Their approach is motivated by the observation that, in the long-time regime, the DSW envelope is typically dominated by a train of solitary waves with a sharp drop to the initial state behind the shock \cite{equalamp,elfocusing}. The method then approximates the DSW as a train of \textit{nearly equal} amplitude solitary waves, with averaging techniques applied to determine the amplitude and velocity of the leading solitary wave edge of the dispersive shock. It is important to note, however, that the DSW equal amplitude approximation does not provide predictions for the trailing edge of a DSW, and extending the work of \cite{equalamp} to address this limitation remains an open problem.

To generate a dispersive shock wave (DSW), one must begin with a rapid change in a physical property, such as an abrupt variation in the surface wave elevation profile $u$. Mathematically, this is modelled by prescribing the conservative-eKdV equation \eqref{e:consekdv} with an initial jump discontinuity,
\begin{equation}\label{e:jump}
   u(x,0)= \left\{ \begin{array}{cc}
         \Delta,~~x<0\\
         0,~~x>0, 
    \end{array}\right.
\end{equation}
\begin{figure}
    \centering
    \includegraphics[width=0.49\textwidth]{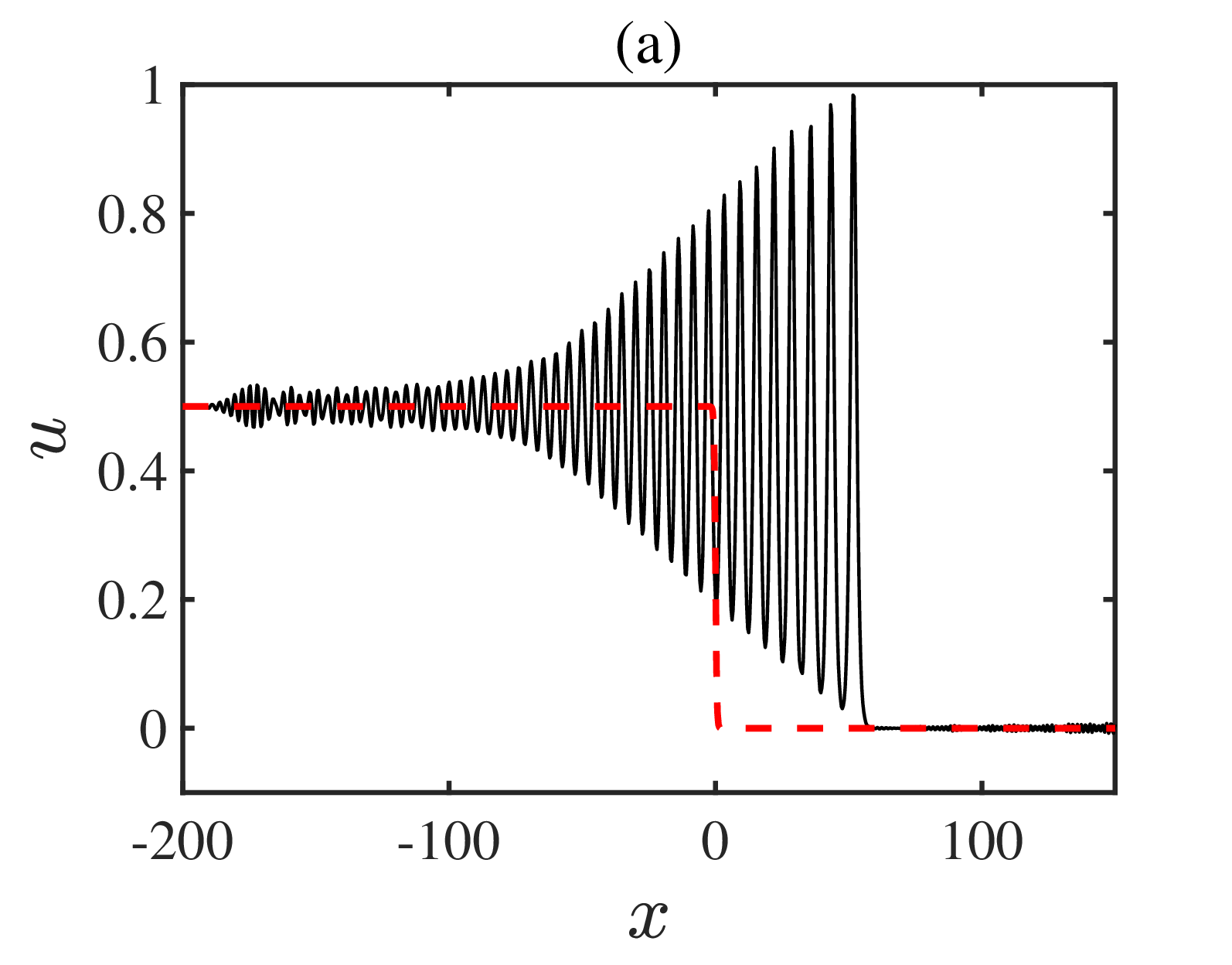}
    \includegraphics[width=0.49\textwidth]{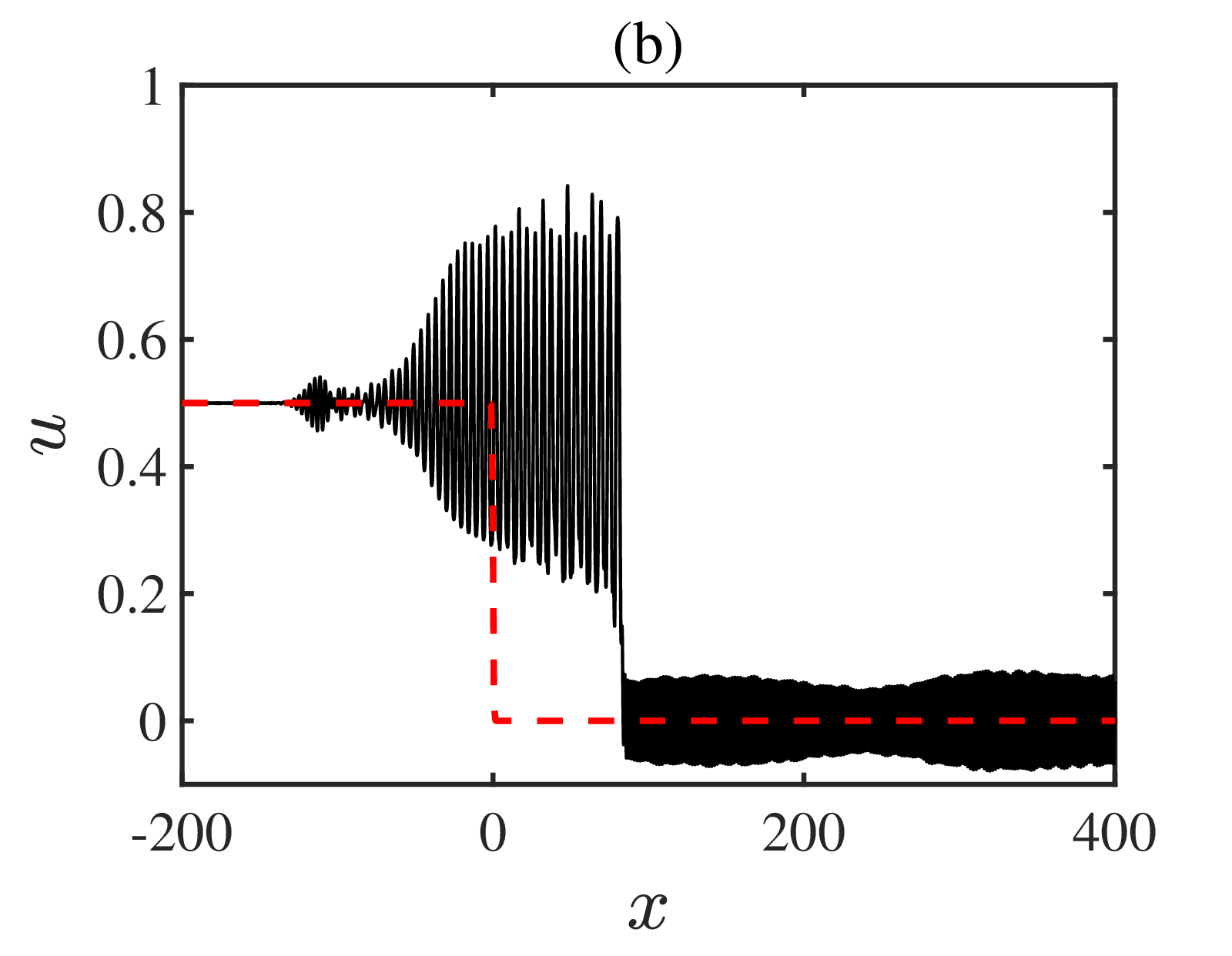}
    \caption{Dispersive hydrodynamic regimes governed by the conservative-eKdV equation (\ref{e:consekdv}). Initial jump discontinuity (\ref{e:jump}): red (dashed) line. (a) Classical form of DSW with $c_{1}=-3/2$, $c_{3}=5/2$, $c_{4}=19/40$, $t=30$, (b) non-classical form of DSW, namely, CDSW regime, with $c_{1}=0.5$, $c_{3}=0.5$, $c_{4}=0.6$, $t=50$. Here, $\epsilon=0.15$, $\Delta=0.5$. (Color version online).}
     \label{f:dsw_regimes}
\end{figure}
which represents a classical Riemann problem. Here, the parameter $\Delta$ specifies the magnitude of the jump. In this section, we study two classes of dispersive shocks, using the variational solitonic solutions (\ref{eq:ws_per}), (\ref{eq:vs_per}), and (\ref{e:Hs_blue}) derived earlier. The first is a classical DSW governed by the conservative-eKdV equation \eqref{e:consekdv}. This regime, illustrated in Figure \ref{f:dsw_regimes}(a), will be analyzed using the DSW equal amplitude approximation. The second is a non-classical (resonant) DSW governed by the conservative-eKdV equation \eqref{e:consekdv}, typically referred to as a CDSW regime. This case, depicted in Figure \ref{f:dsw_regimes}(b), will likewise be studied via the equal amplitude approximation, but in conjunction with the vital concept of Whitham shocks \cite{patjump,salehekdv}.

\subsection{Classical dispersive shock problem} 

Following the approach in \cite{equalamp}, we assume that the dispersive shock wavetrain (illustrated in Figure \ref{f:dsw_regimes}(a)) is predominantly composed of a train of $\mathcal{N}$ solitary waves of nearly uniform amplitude, distributed approximately evenly throughout the bore. To analyze this regime, we begin by integrating the conservative-eKdV mass and energy conservation laws, (\ref{e:ekdvmass}) and (\ref{e:energycons}), over the DSW domain $-\infty < x < \infty$. This procedure leads to the averaged relations
\begin{equation}
   \mathcal{N}\frac{d}{dt}\left(\overline{u}\right)= 3\Delta^{2}+\frac{1}{3}\epsilon c_{1}\Delta^{3},
   \label{e:aveq1}
\end{equation}
\begin{equation}
   \mathcal{N}\frac{d}{dt}\left(\overline{\frac{u^2}{2}}\right)= 2\Delta^{3}+\frac{1}{4}\epsilon c_{1}\Delta^{4},
   \label{e:aveq2}
\end{equation}
\begin{figure}[!ht]
    \centering
    \includegraphics[width=0.49\textwidth]{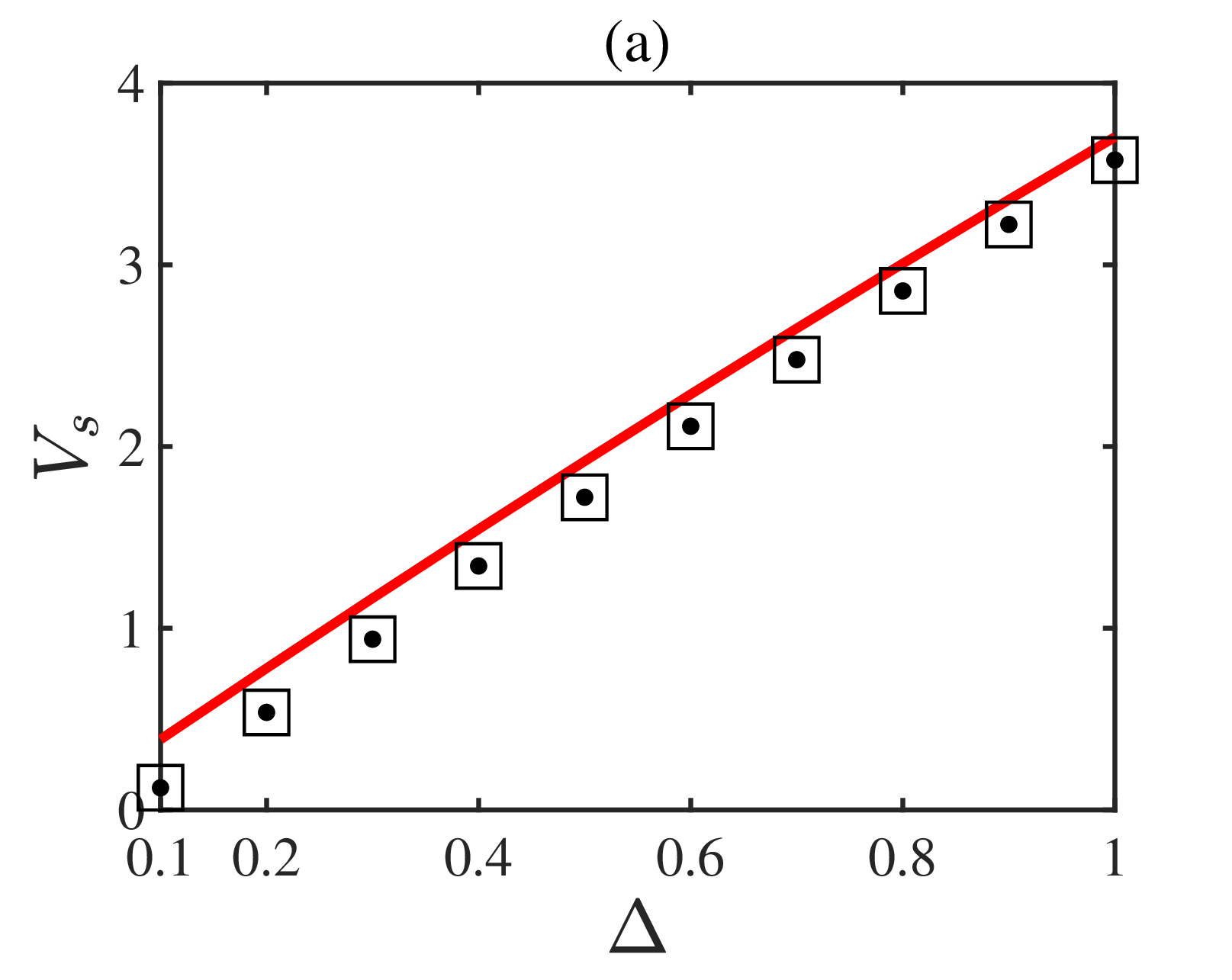}
    \includegraphics[width=0.49\textwidth]{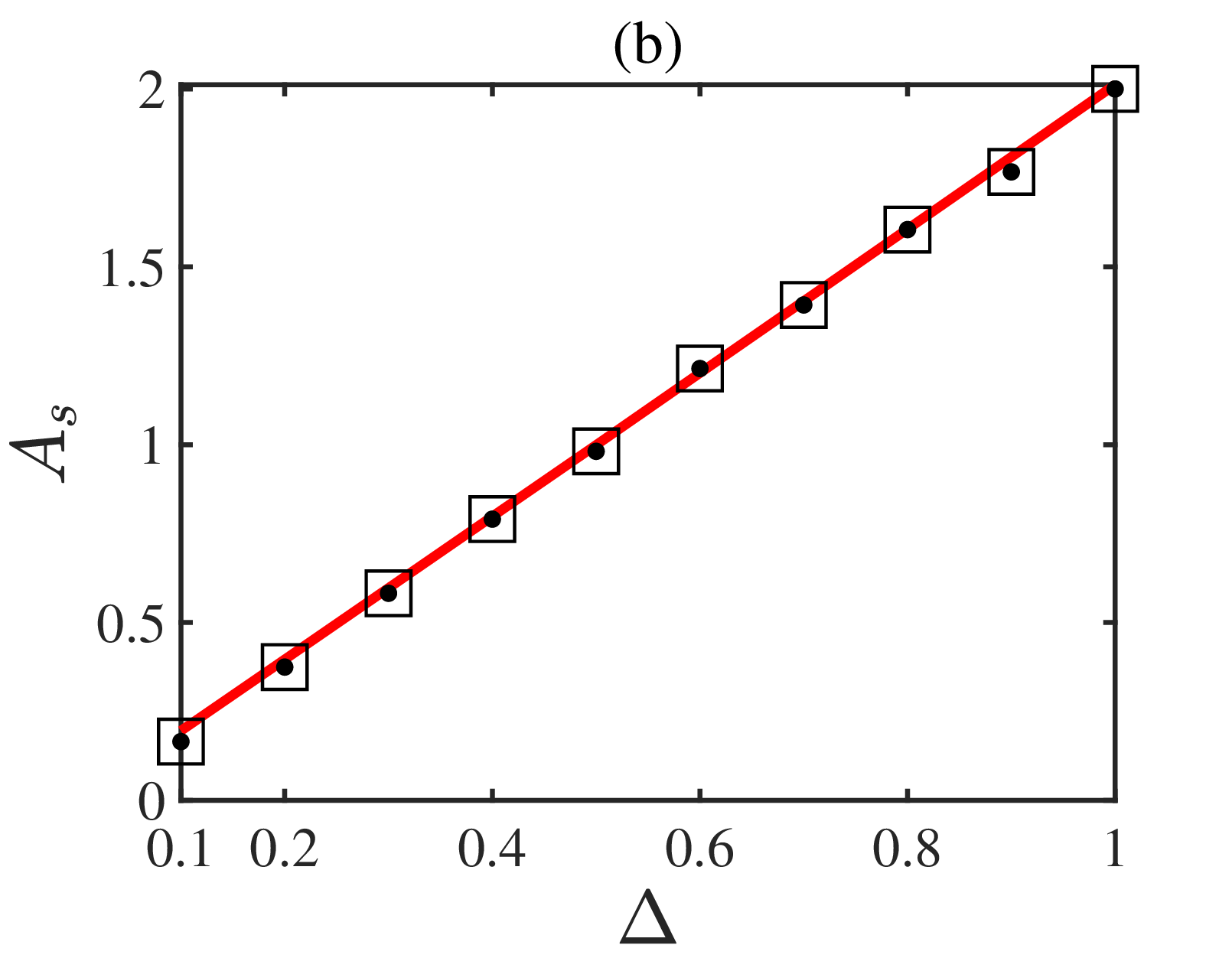}
    \caption{Comparisons between numerical solutions of the conservative-eKdV equation (\ref{e:consekdv}) and analytical solutions obtained by the DSW equal amplitude approximation method. Numerical solutions: black dot-boxes; theoretical solutions: red (solid) line. (a) lead solitary wave edge velocity of the DSW $V_{s}$, (b) lead solitary wave edge height of the DSW $A_{s}$. The higher order coefficients are chosen to correspond to the shallow water wave values: $c_{1}=-3/2$, $c_{3}=5/2$, $c_{4}=19/40$. Here, $t=30$, $\epsilon=0.15$. (Color version online).}
     \label{f:dsw_comps}
\end{figure}expected from the shock evolution at $x= \pm \infty$, which establish a connection between the averaged mass and energy densities. By substituting the variational single solitary wave solutions derived in Section \ref{sec:var}, we obtain the explicit averaged forms
\begin{equation}
    \overline{u}= \int^{\infty}_{-\infty}u\,dx=2\sqrt{2a_{s}}+\epsilon \sqrt{2a^{3}_{s}}\left(-\frac{29}{210}c_{1}+\frac{23}{35}c_{3}-\frac{17}{7}c_{4}\right),
\end{equation}
\begin{equation}
    \overline{\frac{u^2}{2}}=\int^{\infty}_{-\infty}\frac{1}{2}\overline{u^2}\,dx=\frac{2}{3}\sqrt{2a^{3}_{s}}+\epsilon \sqrt{2a^{5}_{s}}\left(-\frac{2}{45}c_{1}+\frac{1}{5}c_{3}-\frac{2}{3}c_{4}\right). 
\end{equation}
Taking the ratio of the averaged mass equation \eqref{e:aveq1} and the averaged energy equation \eqref{e:aveq2}, and integrating with respect to time under the assumption that no solitary waves are initially present in the DSW, yields a closed algebraic relation for the amplitude parameter $a_s$,
\begin{equation}
  \frac{2\sqrt{2a_{s}}+\epsilon \sqrt{2a^{3}_{s}}\left(-\dfrac{29}{210}c_{1}+\dfrac{23}{35}c_{3}-\dfrac{17}{7}c_{4}\right)}{\dfrac{2}{3}\sqrt{2a^{3}_{s}}+\epsilon \sqrt{2a^{5}_{s}}\left(-\dfrac{2}{45}c_{1}+\dfrac{1}{5}c_{3}-\dfrac{2}{3}c_{4}\right) }=\frac{3\Delta^{2}+\dfrac{1}{3}\epsilon c_{1}\Delta^{3}}{2\Delta^{3}+\dfrac{1}{4}\epsilon c_{1}\Delta^{4}},
  \label{e:algebraiceq}
\end{equation}
from which $a_s$ can be determined implicitly using numerical methods such as Newton’s method. The value of this parameter, in turn, directly provides the height $A_{s}$ of the lead solitary wave edge of the DSW via equation (\ref{e:Hs_blue}). Finally, the velocity of the lead solitary wave edge $V_{s}$ is obtained through the solitonic velocity expression (\ref{eq:vs_per}) established earlier. 

To validate the theoretical predictions, we numerically solve the conservative-eKdV equation (\ref{e:consekdv}) by discretizing the spatial domain with a Fourier spectral method and advancing in time with a fourth-order Runge–Kutta (RK4) scheme, following the procedure in \cite{trefethen}. The initial jump discontinuity (\ref{e:jump}) that generates the shock is implemented numerically using an initial tanh-well profile, as described in \cite{wwp}. Figures \ref{f:dsw_comps}(a) and \ref{f:dsw_comps}(b) display comparisons of the predicted and computed values of the lead solitary wave edge velocity and and height of the DSW across a wide range of initial jump magnitudes $\Delta$, respectively. The higher order coefficients $c_{1}$, $c_{3}$, and $c_{4}$ are fixed to the shallow water wave values ($c_{1}=-3/2$, $c_{3}=5/2$, $c_{4}=19/40$) in these comparisons. It can be seen that the agreement between the theoretical and numerical results is excellent.

\subsection{Whitham shock problem}

We next employ the variational single solitary wave solutions to investigate the conservative-eKdV CDSW regime. This non-classical form of resonant DSW has been the subject of extensive study in the contexts of water waves and nonlinear optics \cite{salehnem1,salehnem2,salehekdv}. In the present work, we revisit this problem within the framework of water waves, employing the solutions (\ref{eq:ws_per}), (\ref{eq:vs_per}), and (\ref{e:Hs_blue}) derived above to demonstrate their effectiveness and analytical advantages.

Following \cite{salehnem1,salehnem2,salehekdv}, the resonant wavetrain in front of the bore can be effectively approximated by a uniform Stokes wave oscillating on the background $\bar{u}_{r}$, while the bore itself is modelled as a train of nearly uniform solitary waves using the DSW equal amplitude approximation introduced in the previous subsection. Since the derivation of the CDSW solution parallels the methodology in \cite{salehekdv}, but now under the special conservation case $c_{2}=2c_{3}$ and incorporating the variational single solitary wave solutions, our focus here will be on highlighting the key steps and central results of the derivation.

The Stokes wave approximation of the resonant radiation propagating ahead of the bore is obtained by seeking a weakly nonlinear expansion in the resonant amplitude parameter $a_{r}$, namely, \cite{whitham}
\begin{equation}\label{e:stokes}
    u_{r} =  \bar{u}_{r}+a_{r}\cos(\theta_{r})+a^{2}_{r}u_{2}\cos(2\theta_{r})+\mathcal{O}(a^3_r), 
\end{equation}
with a corresponding weakly nonlinear expansion for the resonant frequency,
\begin{equation}
\omega_{r}=\omega_{0}+a_{r}\omega_{1}+a^{3}_{r}\omega_{2}+\mathcal{O}(a^3_r),  
 \label{e:stokes_gen_w}
\end{equation}
where $\theta_{r}=k_{r}x-\omega_{r}t$ denotes the resonant phase. Applying a standard asymptotic procedure to determine the unknown Stokes coefficients in the expansions (\ref{e:stokes}) and (\ref{e:stokes_gen_w}) by substituting into the conservative-eKdV equation \eqref{e:consekdv} yields
\begin{eqnarray}
& & \omega_{0} = (6\bar{u}_{r}+\epsilon c_{1}\bar{u}^{2}_{r})k_{r}-\left(1+\epsilon c_{3}\bar{u}_{r}\right)k^3_{r} + \epsilon c_{4} k^5_{r}, \label{e:omega0} \\
& &  \omega_{2} = \frac{36+24\epsilon c_1\bar{u}_{r} - \epsilon \left(36c_3-6c_1\right)k^{2}_{r}}{24k_{r} - \epsilon\left(120c_4k^3_{r}-24c_3\bar{u}_{r}k_{r}\right),}
\nonumber \\
& & \hspace{0.18in}= \frac{3}{2k_{r}}+\epsilon\left(\frac{1}{4}\left( c_{1}-6c_{3}+ 30 c_{4}\right)k_{r} + \left( c_{1} - \frac{3}{2}c_{3} \right) \frac{\bar{u}_{r}}{k_{r}}\right)+\mathcal{O}\left(\epsilon^2\right), \label{e:omega2}  \\
& &  u_{2} = \frac{6+2\epsilon c_{1}\bar{u}_{r}-3\epsilon c_{3} k^2_{r}}{12k^2_{r} + 12\epsilon\left(c_{3}\bar{u}_{r} - 5c_4k^2_{r}\right)k^{2}_{r}}  
\nonumber \\
& & \hspace{0.18in}=\frac{1}{2k^2_{r}} - \frac{1}{12}\epsilon\left( 3c_3 - 30 c_{4} - 2\left(c_{1} - 3c_{3} \right) \frac{\bar{u}_{r}}{k^2_{r}}\right) + \mathcal{O}\left(\epsilon^2\right), \label{e:u2}
\end{eqnarray}
with $\omega_{1}=0$, after removal of secular terms throughout the asymptotic procedure \cite{whitham,kamchatnovbook}. The Stokes wave must then be connected to the unstable bore behind via modulation jump conditions. This jump system is known as a Whitham shock system—a concept in dispersive hydrodynamics originally speculated by Whitham \cite{whitham} but not pursued, and later formally developed and pioneered by Hoefer \textit{et al.} \cite{patjump}.

To carry out the Whitham shock calculations, we impose moving jumps in the mean level $\bar{u}$, amplitude $a$, and wavenumber $k$ modulation variables,
\begin{equation}
   \bar{u}(x,t)= \left\{ \begin{array}{cc}
         \bar{u}_{s},~~x<U_{s}t\\
         u_{+},~~x>U_{s}t
    \end{array}\right., \quad 
    a(x,t)= \left\{ \begin{array}{cc}
         a_{s},~~x<U_{s}t\\
         a_{r},~~x>U_{s}t
    \end{array}\right., \\ 
    \quad k(x,t)= \left\{ \begin{array}{cc}
         0,~~x<U_{s}t\\
         k_{r},~~x>U_{s}t
    \end{array}\right. ,
\end{equation}
which can be viewed as the dispersive analogues of the classical Rankine–Hugoniot jump conditions. Here, $\bar{u}_{s}$ denotes the background of the lead solitary wave edge, and $U_{s}$ represents the Whitham shock velocity. In these modulation jump conditions, the level behind the jump in the wavenumber $k$ is set to zero, since the solitary wave corresponds to an infinite wavelength and thus has a wavenumber near zero \cite{whitham,kamchatnovbook}. Moreover, the level ahead in the jump for the mean $\bar{u}$ is approximated by $u_{+}$ (the initial level ahead of the shock), as numerical evidence shows that the mean of the resonant radiation is extremely close to $u_{+}$ \cite{pat,patjump,salehekdv}, rendering their difference negligible. This approximation also ensures that the number of wave parameter unknowns matches the number of equations that can be derived. Altogether, the system involves five unknown parameters: the background of the lead solitary wave edge $\bar{u}_{s}$, the solitary wave edge amplitude parameter $a_{s}$, the resonant amplitude parameter $a_{r}$, the resonant wavenumber $k_{r}$, and the Whitham shock velocity $U_{s}$.

The travelling discontinuities introduced above must now be imposed as shock conditions for the mass and energy conservation laws of the conservative-eKdV equation \eqref{e:ekdvmass} and \eqref{e:energycons}. This yields the Whitham modulation jump conditions
\begin{eqnarray}
    -U_{s}\left(\overline{\mathcal{D}}_{\text{m,bore}}-\overline{\mathcal{D}}_{\text{m,Stokes}}\right)+\left(\overline{\mathcal{F}}_{\text{m,bore}}-\overline{\mathcal{F}}_{\text{m,Stokes}}\right) & = &0,
    \label{e:wjump1} \\
    -U_{s}\left(\overline{\mathcal{D}}_{\text{e,bore}}-\overline{\mathcal{D}}_{\text{e,Stokes}}\right)+\left(\overline{\mathcal{F}}_{\text{e,bore}}-\overline{\mathcal{F}}_{\text{e,Stokes}}\right) & = &0,
    \label{e:wjump2}
\end{eqnarray}
where
\begin{eqnarray}
  \overline{\mathcal{D}}_{m} & = & \overline{u}, \label{e:dm} \\
  \overline{\mathcal{D}}_{e} & = & \frac{1}{2}\overline{u^{2}},
  \label{e:de} \\
  \overline{\mathcal{F}}_{m} & = & 3\overline{u^{2}} + \overline{u_{xx}} + \epsilon \left( \frac{1}{3} c_{1} \overline{u^{3}} + c_{3}\overline{uu_{xx}} + \frac{1}{2} c_{3} \overline{u_{x}^{2}} + c_{4} \overline{u_{xxxx}} \right), \label{e:fm} \\
  \overline{\mathcal{F}}_{e} & = & 2\overline{u^{3}} + \overline{uu_{xx}}
 - \frac{1}{2}\overline{u_{x}^{2}} + \epsilon \left( \frac{1}{4} c_{1} \overline{u^{4}} + c_{3} \overline{u^{2}u_{xx}} + c_{4} \overline{uu_{xxxx}} \right.
 \nonumber \\
 & & \left. \mbox{}- c_{4} \overline{u_{x}u_{xxx}} + \frac{1}{2} c_{4} \overline{u_{xx}^{2}} \right). \label{e:fe}
\end{eqnarray}
The symbols $\mathcal{D}$ and $\mathcal{F}$ denote the density and flux terms in the conservation laws, respectively, while the subscripts `m' and `e' correspond to the mass and energy equations. The averaging rules used in \eqref{e:wjump1} and \eqref{e:wjump2} are specified by
\begin{equation}
    \overline{f}_{\text{Stokes}}=\frac{1}{2\pi}\int^{2\pi}_{0} f(u,u_{\theta},u_{\theta\theta},\ldots) \,d\theta,
\end{equation}
when $u$ is taken as the Stokes wave solution $u=u_{r}$ \eqref{e:stokes}, and by
\begin{equation}
    \overline{f}_{\text{bore}}=\int^{\infty}_{-\infty} f(u,u_{\theta},u_{\theta\theta},\ldots) \,d\theta,
\end{equation}
when $u$ is the variational single solitary wave solution (\ref{e:ansatz}) with the wave parameter expressions \eqref{eq:ws_per}, \eqref{eq:vs_per}, and \eqref{e:Hs_blue}. In this setting, however, the background associated with the solitary wave must be explicitly included in the solution representation, which takes the form
\begin{equation}
u_{s} = \bar{u}_{s} + A_{s}\sech^2(w_{s}\,\theta_{s}).
\end{equation}
The role of the background level at the lead solitary wave edge is particularly important in the Whitham jump calculations. Physically, the resonant radiation gradually lifts the initial upstream state $u_{+}$ to the mean level of the resonant wavetrain $\bar{u}_{r}$. This mean value then undergoes a discontinuous transition through the Whitham shock to the bore’s mean value $\bar{u}_{s}$, thereby providing the link between the resonant wavetrain and the unstable undular bore. A detailed theoretical and numerical justification of this mechanism is given in \cite{salehekdv}.

The Whitham modulation jump conditions \eqref{e:wjump1}–\eqref{e:wjump2} enable the determination of two unknown wave parameters. To close the system, three additional equations are required. The third equation follows from applying the DSW equal amplitude approximation to the conservative-eKdV bore, as in the earlier application, since the bore can be effectively approximated by a train of nearly equal amplitude solitary waves. This procedure again produces an implicit relation for the amplitude parameter $a_{s}$, but here it must be modified to account for the resonant radiation ahead of the bore, modelled by the Stokes wave \eqref{e:stokes}, together with the background level of the leading solitary wave $\bar{u}_{s}$. These extensions yield the implicit equation
\begin{equation}
     \frac{\bar{u}_{s}+2\Psi_{1}}{\dfrac{1}{2}\bar{u}_{s}^{2}+2\bar{u}_{s} \Psi_{1} +\dfrac{2}{3}\Psi_{2}} \\ 
     =\frac{\left(3\Delta^{2}+\dfrac{1}{3}\epsilon c_{1}\Delta^{3}\right)-\mathcal{F}_{\text{m,Stokes}}}{\left(2\Delta^{3}+\dfrac{1}{4}\epsilon c_{1}\Delta^{4}\right)-\mathcal{F}_{\text{e,Stokes}}},
\end{equation}
with
\begin{eqnarray}
\Psi_{1} & = & \sqrt{2a_{s}}+\frac{1}{2}\epsilon \sqrt{2a^{3}_{s}}\left(-\frac{29}{210}c_{1}+\frac{23}{35}c_{3}-\frac{17}{7}c_{4}\right), \\
 \Psi_{2} & = & \sqrt{2a^{3}_{s}}+\frac{3}{2}\epsilon \sqrt{2a^{5}_{s}}\left(-\frac{2}{45}c_{1}+\frac{1}{5}c_{3}-\frac{2}{3}c_{4}\right). 
\end{eqnarray}
As the bore in the CDSW regime is approximated by a nearly uniform solitary wavetrain, the fourth equation is provided by the velocity of the variational solitary wave solution (\ref{eq:vs_per}), which here serves as the Whitham shock velocity, namely, 
\begin{equation}
    U_{s}=V_{s}=2a_{s}+\epsilon a^{2}_{s}\left(\frac{6}{35}c_{1} - \dfrac{2}{35} c_3-\frac{4}{7}c_{4}\right).
\end{equation}
Finally, the system is closed with the resonance condition: the leading solitary wave edge of the bore must propagate at the phase velocity of the resonant radiation, that is, 
\begin{equation}
    U_{s}=\frac{\omega_{r}}{k_{r}}=\frac{\omega_{0}+a^{2}_{r}\omega_{2}}{k_{r}},
\end{equation}
This condition is equivalent to the modulation jump relation associated with the wave conservation equation $k_{t}+\omega_{x}=0$ \cite{wwp}.

To validate the theory, we compare with direct numerical simulations of the conservative-eKdV equation in the CDSW regime. Two sets of runs are performed: (i) fixing $c_{1}$ and $c_{3}$ while varying $c_{4}$, and (ii) fixing $c_{4}$ while varying $c_{1}$ and $c_{3}$. For simplicity, we always set $c_{1}=c_{3}$, and $c_{2}=2c_{3}$ as required by energy conservation.
\begin{figure}[!ht]
    \centering
    \includegraphics[width=0.49\textwidth]{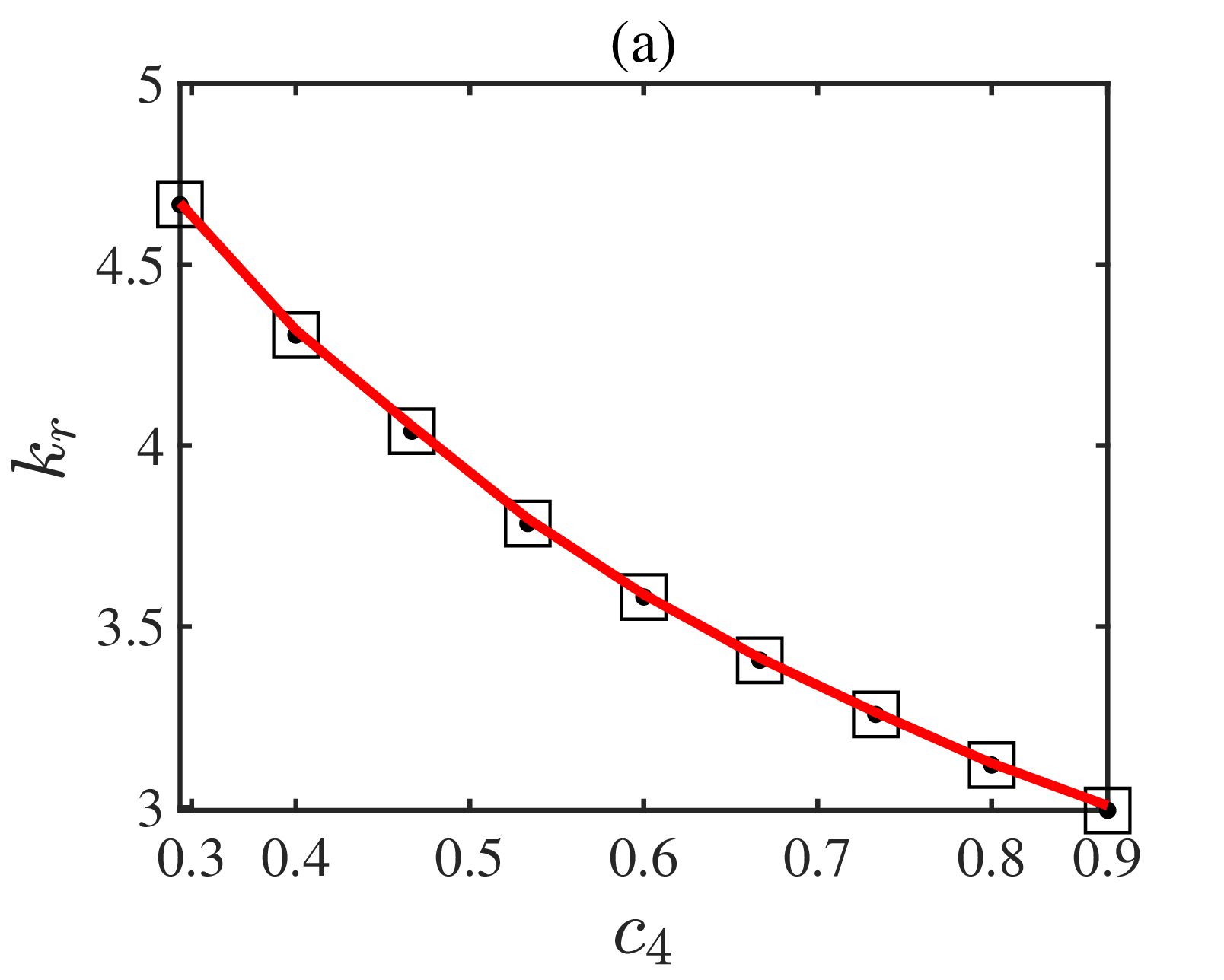}
    \includegraphics[width=0.49\textwidth]{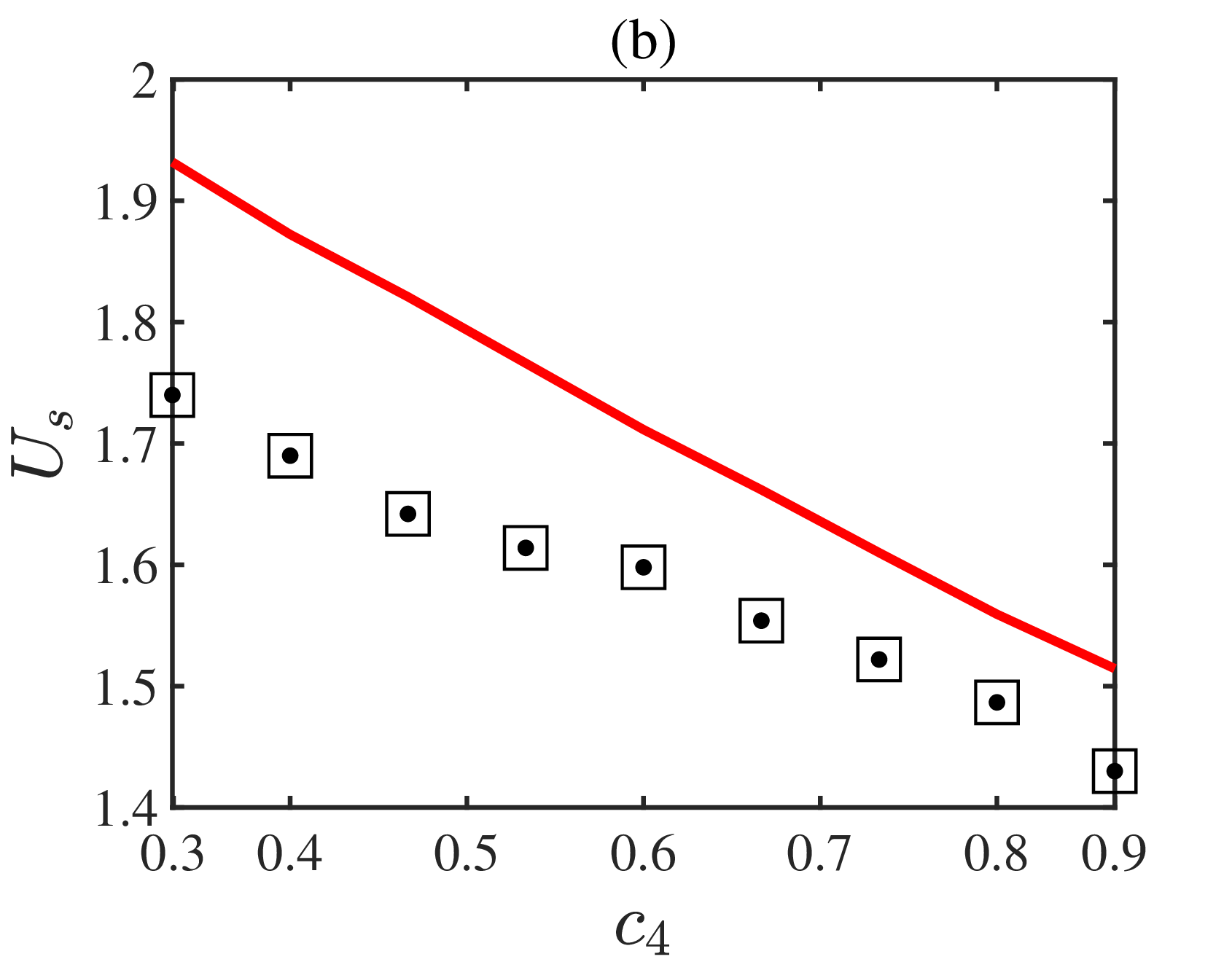}
    \includegraphics[width=0.49\textwidth]{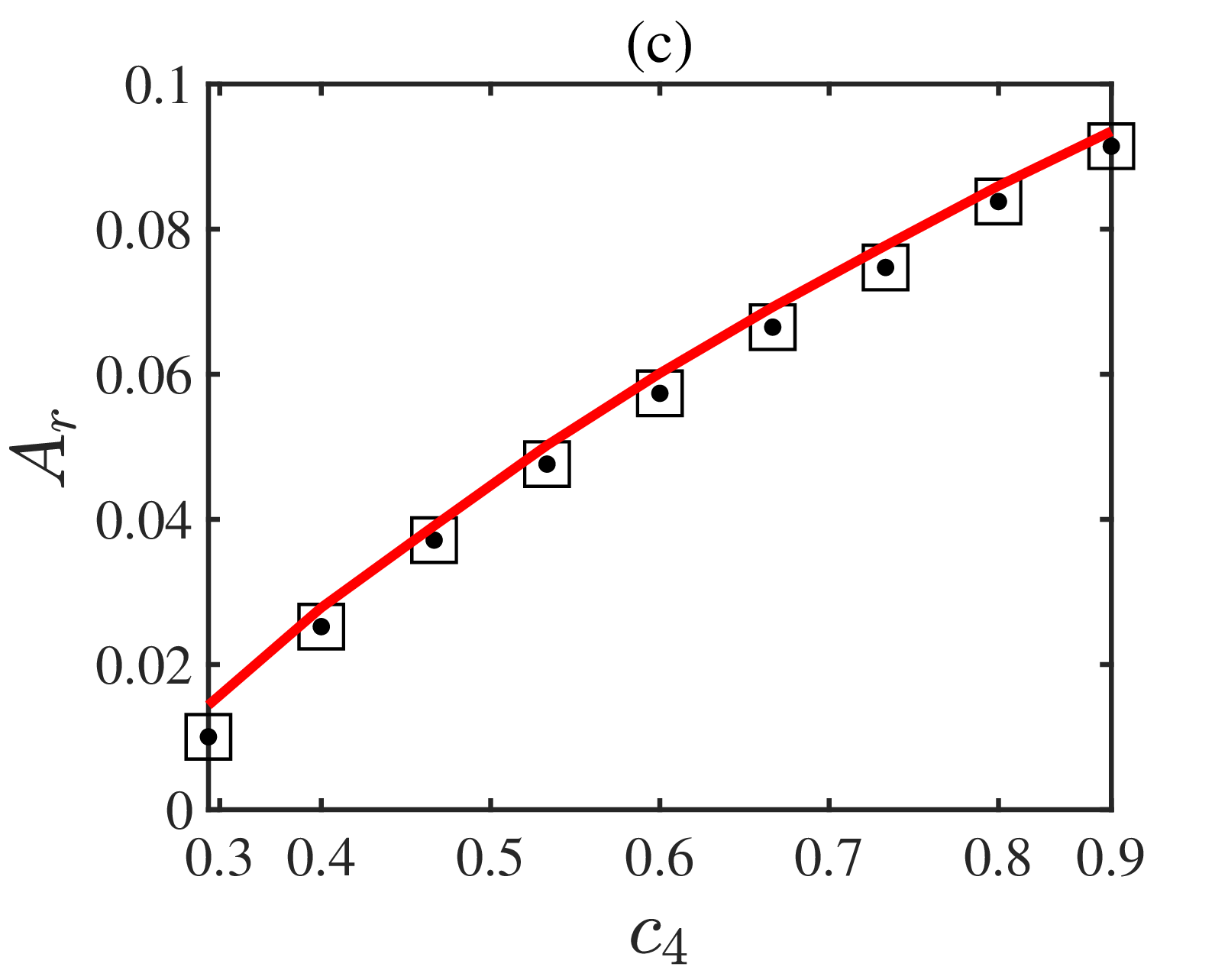}
    \includegraphics[width=0.49\textwidth]{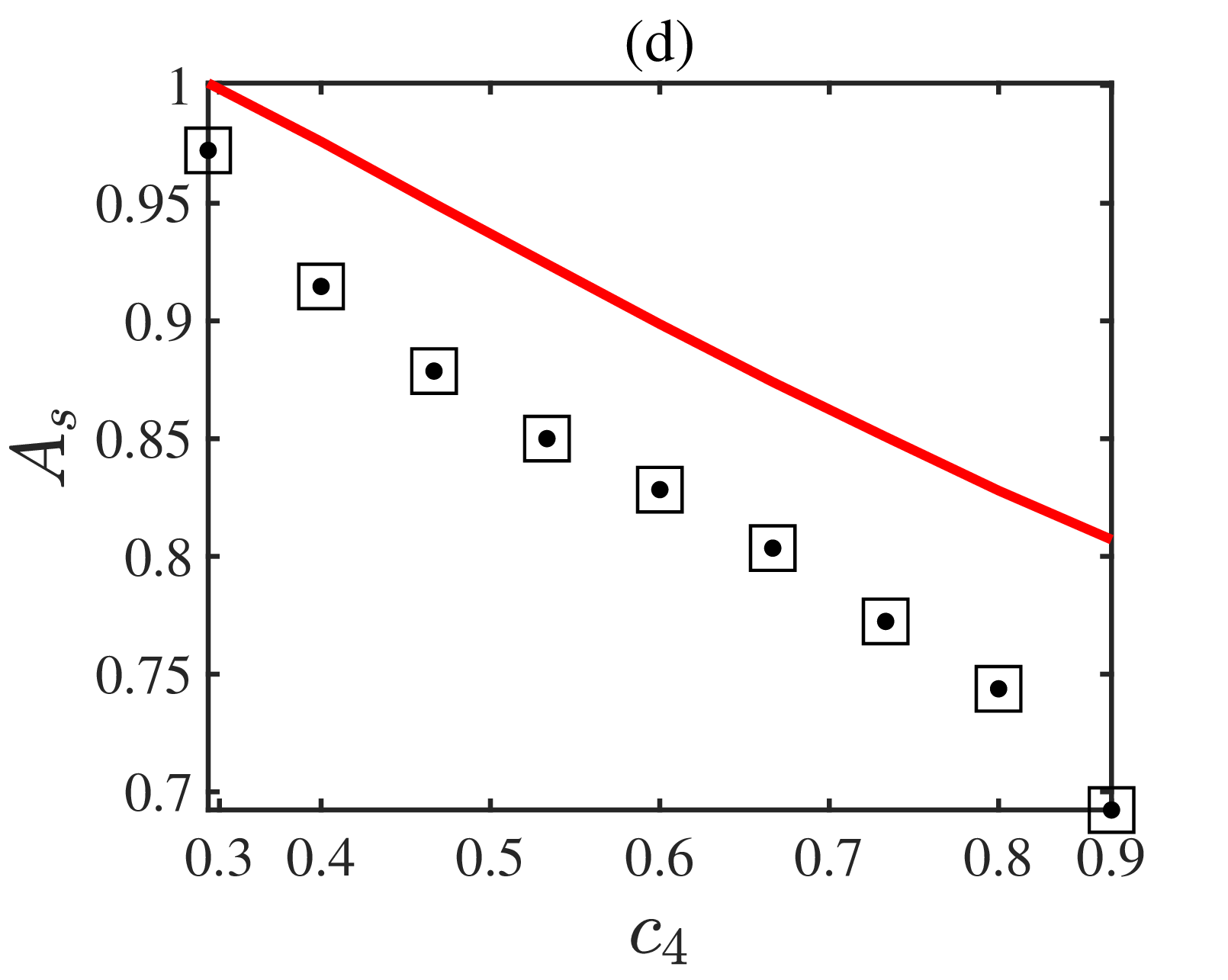}
    \caption{Comparisons between numerical solutions of the conservative-eKdV equation (\ref{e:consekdv}) and analytical solutions obtained by the Whitham shock method for the CDSW regime. Numerical solutions: black dot-boxes; theoretical solutions: red (solid) line. (a) resonant wavenumber of radiation $k_{r}$, (b) velocity of Whitham shock $U_{s}$, (c) resonant height of radiation $A_{r}$, (d) lead solitary wave edge height of the bore $A_{s}$. Here, $c_{1}=c_{3}=0.5$, $\Delta=0.5$, $u_{+}=0$, $t=20$, $\epsilon=0.15$. (Color version online).}
     \label{f:cdsw_varyc4}
\end{figure}
Figures \ref{f:cdsw_varyc4}(a) to \ref{f:cdsw_varyc4}(d) present comparisons between theoretical predictions and numerical results for the resonant wavenumber $k_{r}$, the Whitham shock velocity $U_{s}$, the resonant wave height $A_{r}$, and the lead solitary wave edge height $A_{s}$, respectively. In these runs, the coefficients are chosen as $c_{1}=c_{3}=0.5$ and $\epsilon=0.15$, while $c_{4}$ is varied. The resonant wavenumber $k_{r}$ is extracted numerically from the mean resonant wavelength $\lambda_{r}$ averaged over 20--30 wave crests, via $k_{r}=2\pi/\lambda_{r}$. It is important to note that the resonant wave height $A_{r}$ differs from the resonant amplitude parameter $a_{r}$, since the actual height is determined from the maximum of the Stokes wave solution \eqref{e:stokes}, yielding $A_{r}=u_{+}+a_{r}+a_{r}^{2}u_{2}$, where $a_{r}$ is computed numerically from the Whitham shock system. Similarly, the lead solitary wave edge height is given by $A_{s}=a_{s}+\epsilon a^{2}_{s}\left(c_{1}/420 -c_3/35 + 3c_{4}/14\right)$, as derived from the variational single solitary wave solution in Section \ref{sec:var}, with $a_{s}$ again determined numerically from the Whitham system. Across all comparisons, the agreement between theory and numerics is excellent, with maximum errors of 19.95\% and 11.48\% in $U_{s}$ and $A_{s}$, respectively.

A second set of comparisons, shown in Figures \ref{f:cdsw_fixc4}(a) to \ref{f:cdsw_fixc4}(d), fixes $c_{4}$ at its shallow water wave value $c_{4}=19/40$ while varying $c_{1}=c_{3}$, again with $\epsilon=0.15$. As before, the theoretical predictions closely match the numerics, with maximum errors of 11.5\% and 9.07\% in the Whitham shock velocity $U_{s}$ and the lead solitary wave edge height $A_{s}$, respectively.

\begin{figure}[!ht]
    \centering
    \includegraphics[width=0.49\textwidth]{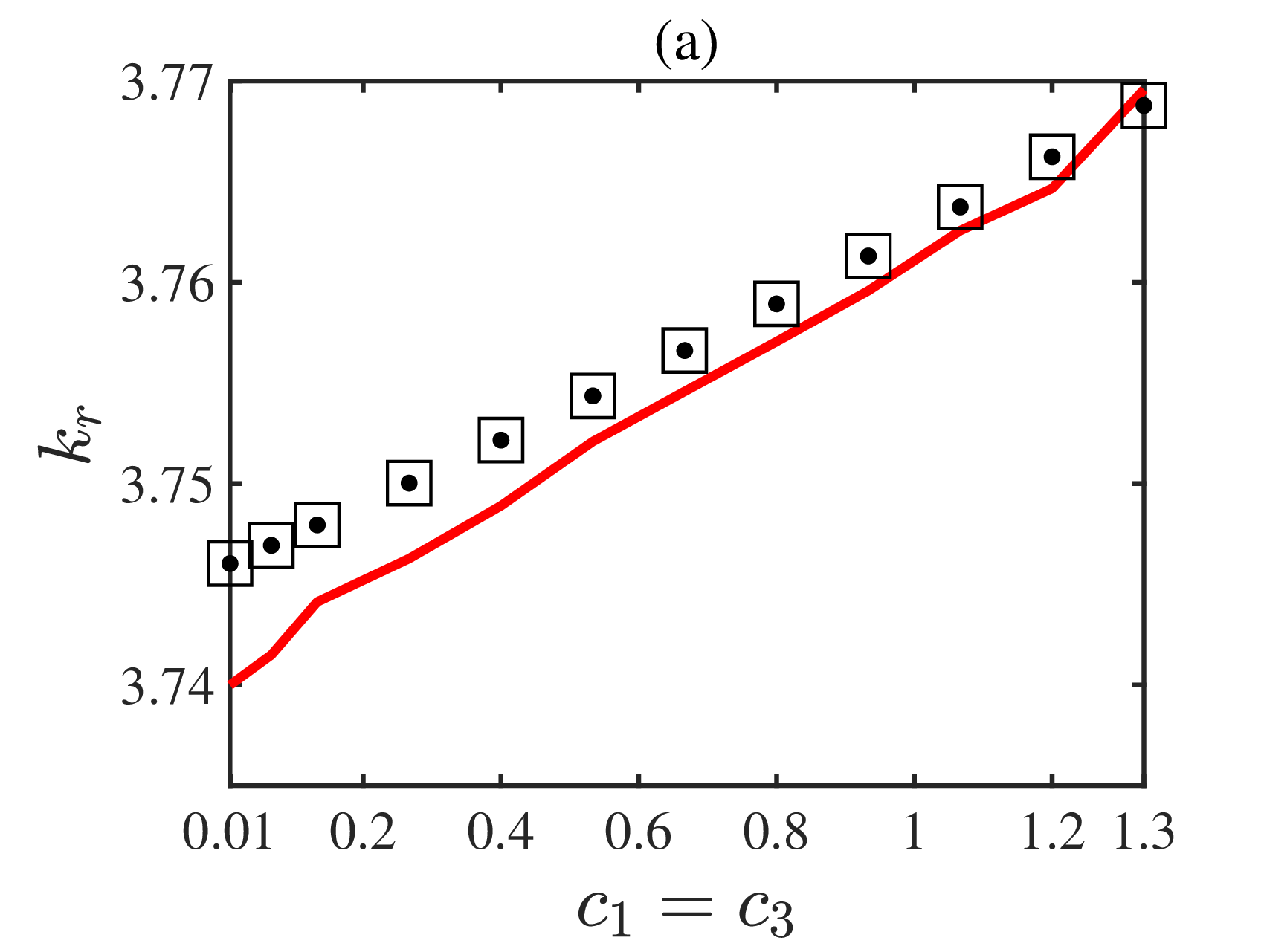}
    \includegraphics[width=0.49\textwidth]{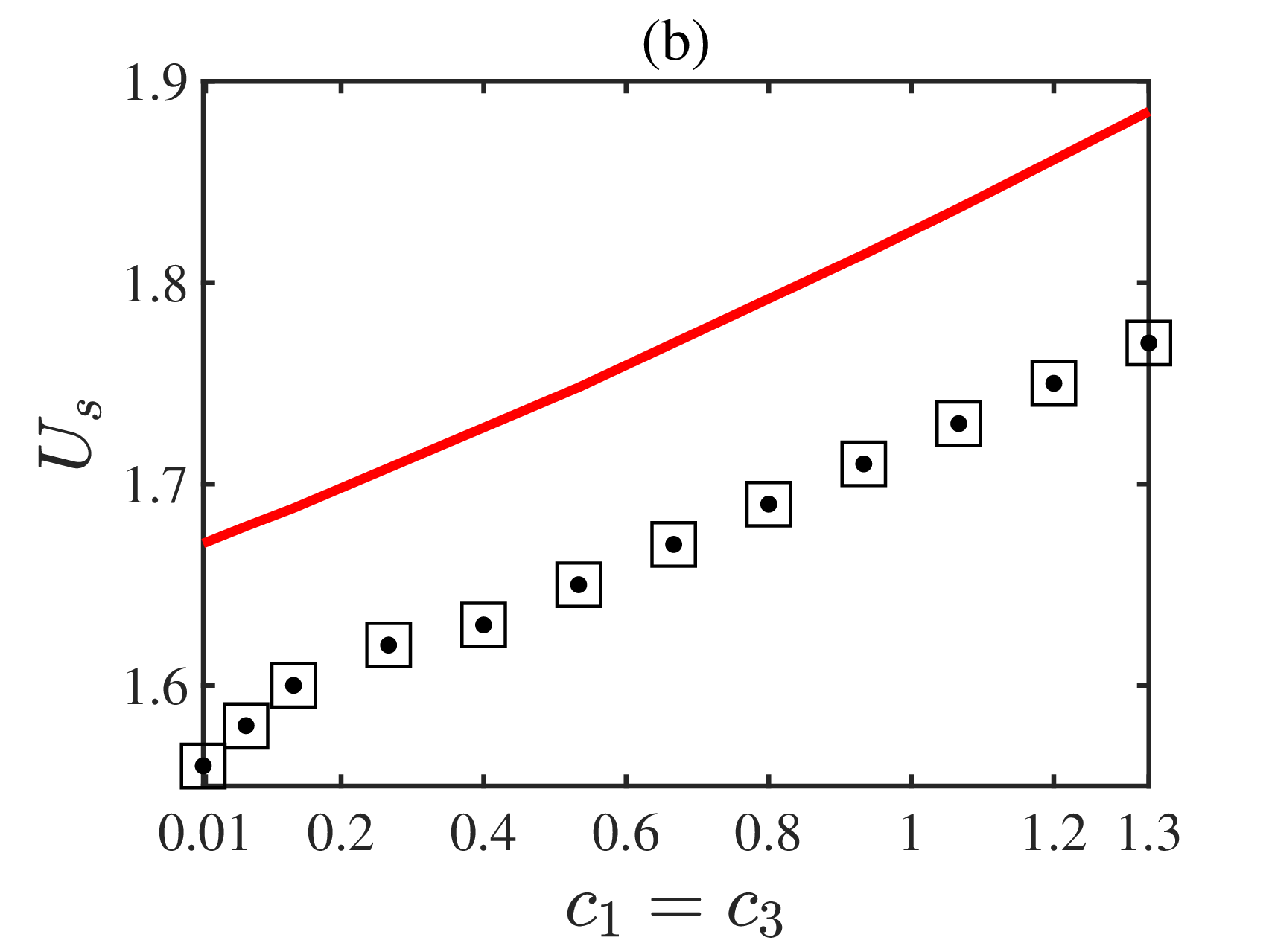}
    \includegraphics[width=0.49\textwidth]{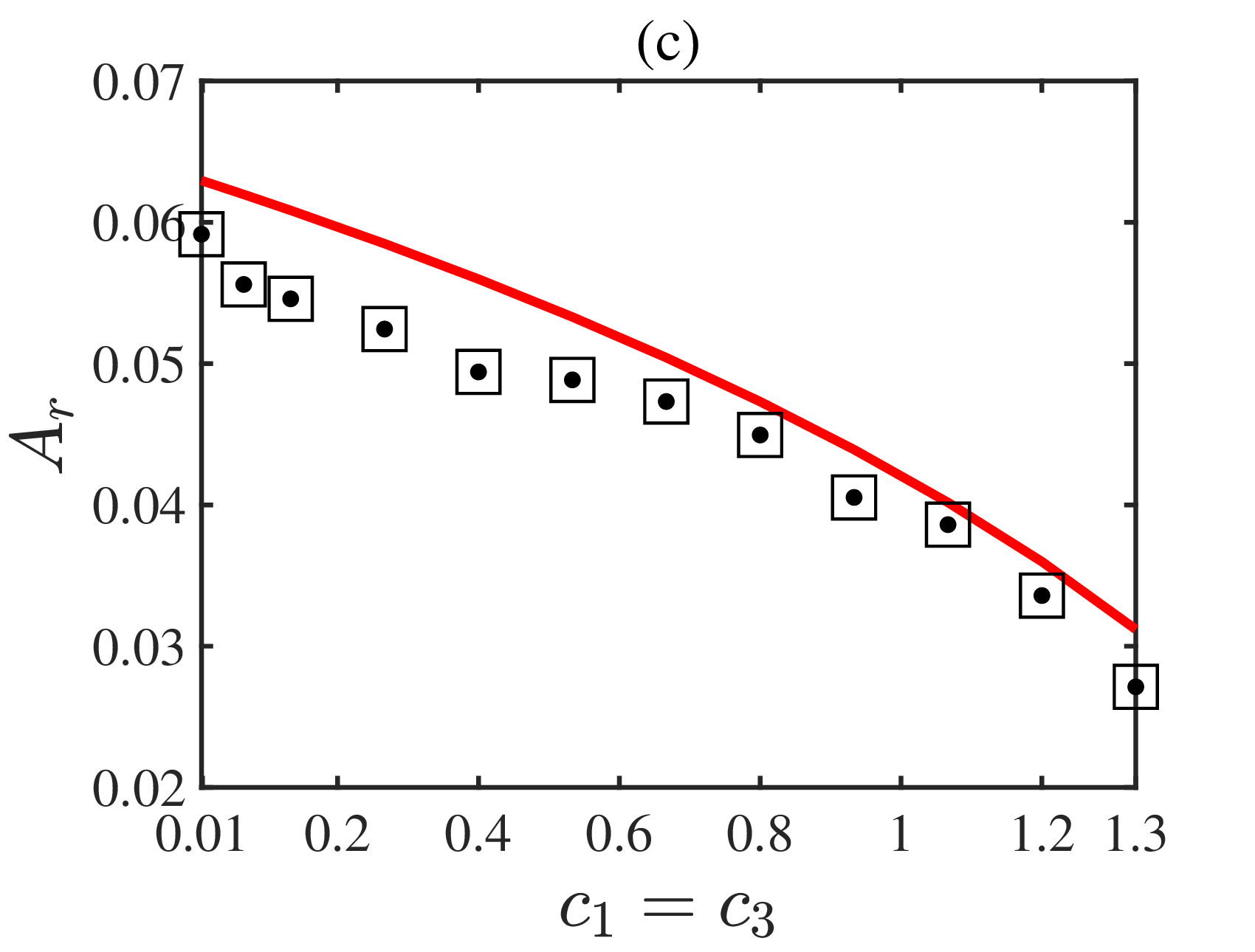}
    \includegraphics[width=0.49\textwidth]{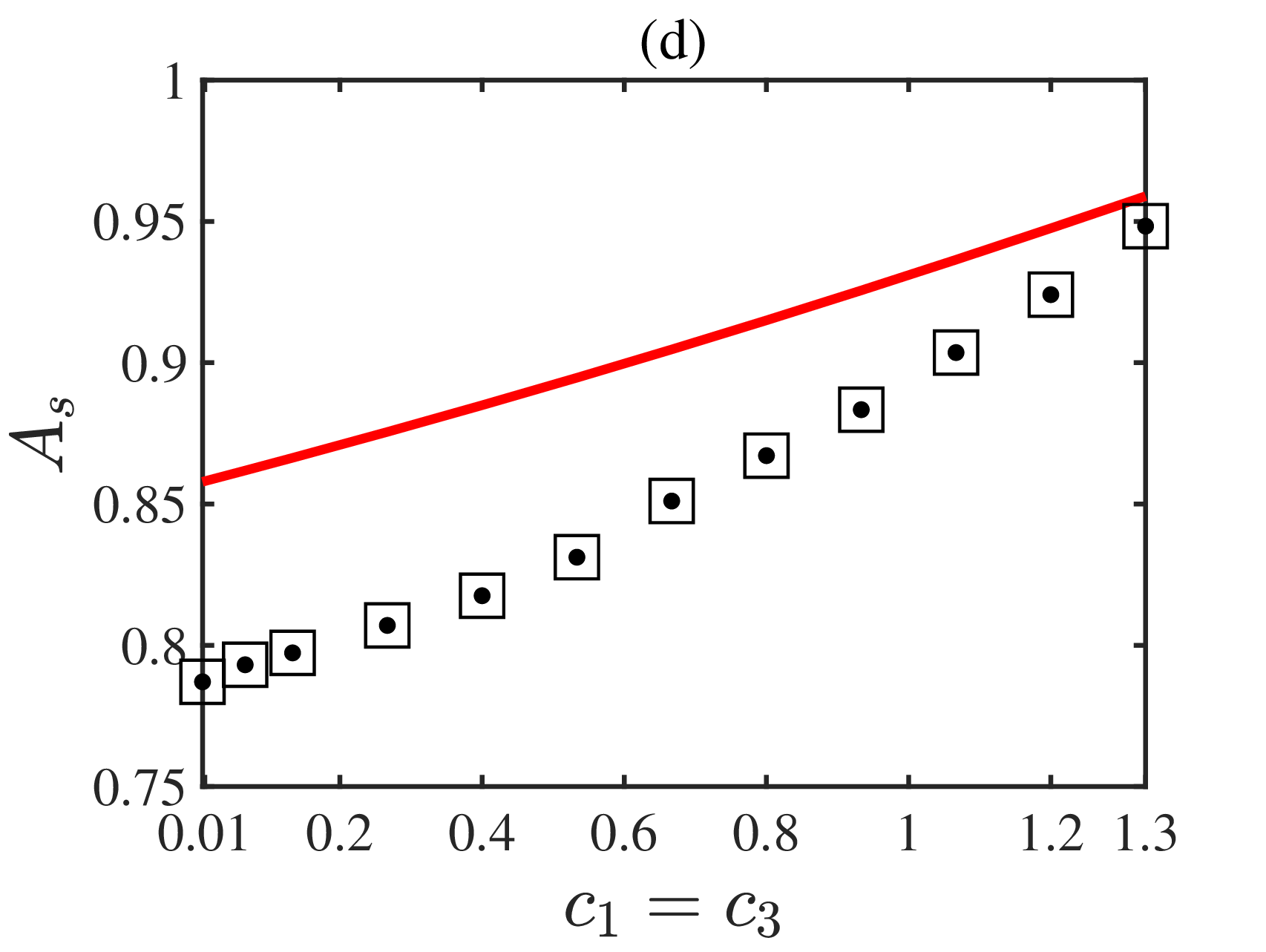}
    \caption{Comparisons between numerical solutions of the conservative-eKdV equation (\ref{e:consekdv}) and analytical solutions obtained by the Whitham shock method for the CDSW regime. Numerical solutions: black dot-boxes; theoretical solutions: red (solid) line. (a) resonant wavenumber of radiation $k_{r}$, (b) velocity of Whitham shock $U_{s}$, (c) resonant height of radiation $A_{r}$, (d) lead solitary wave edge height of the bore $A_{s}$. Here, $c_{4}=19/40$, $\Delta=0.5$, $u_{+}=0$ $t=20$, $\epsilon=0.15$. (Color version online).}
     \label{f:cdsw_fixc4}
\end{figure}

\section{Conclusion and future directions}\label{s:conc}

%In this work, we have established the existence of a variational single solitary wave solution for a particular class of the extended KdV (eKdV) equation \eqref{e:ekdv} in which energy conservation plays a central role, namely the case $c_{2}=2c_{3}$. It is shown that this model, referred to throughout as the conservative-eKdV equation (\ref{e:consekdv}), admits a single solitary wave solution in the form of $\sech^{2}$ profile using rigorous tools from the calculus of variations, specifically the averaged Lagrangian method. By formulating the associated Lagrangian and Hamiltonian structures, the variational framework yields explicit expressions, in terms of the solitonic amplitude parameter $a_{s}$, for the higher order corrections to the solitary wave velocity $V_{s}$ and (inverse) width $w_{s}$, and—when combined with basic energy estimates—for the solitary wave height $A_{s}$.

%Hamid's suggestion (bit of rephrasing) 
In this work, we have established the existence of a single solitary wave solution for a class of extended KdV (eKdV) equations \eqref{e:ekdv} that satisfy energy conservation, namely when $c_{2}=2c_{3}$. It is shown that this model, referred to throughout as the conservative-eKdV equation (\ref{e:consekdv}), admits a single solitary wave solution in the form of $\sech^{2}$ profile using tools from the calculus of variations, specifically the averaged Lagrangian method. By formulating the associated Lagrangian and Hamiltonian structures the variational framework yields explicit expressions, in terms of the solitonic amplitude parameter $a_{s}$, for the higher order corrections to the solitary wave velocity $V_{s}$ and (inverse) width $w_{s}$, and—when combined with basic energy estimates—for the solitary wave height $A_{s}$.

Compared to earlier studies of the eKdV equation that relied primarily on algebraic and asymptotic techniques, the variational solitonic solutions derived here are notably simpler in form and considerably more convenient for applications within Whitham modulation theory. This analytical tractability makes them particularly well suited for studying nonlinear wave phenomena in dispersive hydrodynamics.

To demonstrate the accuracy and effectiveness of the proposed solutions, two distinct applications were examined using tools from dispersive hydrodynamics, including the DSW equal amplitude approximation and the concept of Whitham shocks. The first application concerned the classical formation of dispersive shock waves governed by the conservative-eKdV equation \eqref{e:consekdv}. The second addressed a non-classical dispersive shock regime, referred to as the Cross-over DSW (CDSW), in which the leading solitary wave edge of the bore interacts resonantly with radiation propagating ahead of the bore, while the interior wavetrain of the bore becomes modulationally unstable. In both cases, detailed comparisons with direct numerical simulations of the conservative-eKdV equation show excellent agreement between theoretical predictions and numerical results.

Several open problems naturally arise from this study. One particularly challenging and interesting direction is the extension of the present analysis to the full eKdV equation in which energy conservation is no longer enforced, that is, when $c_{2}\neq 2c_{3}$. This regime presents significant analytical difficulties, as no direct Lagrangian or Hamiltonian structures governing the eKdV equation is currently available to support the construction of variational single solitary wave solutions. Nevertheless, this problem is the subject of ongoing investigation, and preliminary results indicate promising progress, particularly in the context of Whitham modulation theory and its applications to non-integrable dispersive systems.

\end{document}